\documentclass[amssymb, nofootinbib, superscriptaddress]{revtex4}
\usepackage{amsfonts,amssymb,amsmath,exscale,bbm}
\usepackage{footnpag} 
\usepackage[all,knot]{xy}  
\usepackage{color}
\usepackage{graphicx}
\usepackage{epsfig}
\usepackage{subfig}
\usepackage{enumerate}
\usepackage{bbold}

\usepackage{amsmath}
\usepackage{amsfonts}
\usepackage{graphicx}
\usepackage{psfrag}
\usepackage{array}
\usepackage{bbm}
\usepackage{amsthm}
\theoremstyle{definition}
\newtheorem{definition}{Definition}

\newtheorem{proposition}{Proposition}
\newtheorem{corollary}{Corollary}

\newcommand{\cA}{{\mathcal A}}
\newcommand{\cB}{{\mathcal B}}

\newcommand{\cG}{{\mathcal G}}

\newcommand{\cM}{{\mathcal M}}

\newcommand{\cO}{{\mathcal O}}

\newcommand\beq{\begin{equation}}
\newcommand\eeq{\end{equation}}
\newcommand{\be}{\begin{equation}}
\newcommand{\ee}{\end{equation}}
\newcommand{\bes}{\begin{eqnarray}}
\newcommand{\ees}{\end{eqnarray}}
\newcommand{\bea}{\begin{eqnarray}}
\newcommand{\eea}{\end{eqnarray}}

\def\vphi{{\varphi}}

\newcommand{\U}{\mathrm{U}}

\newcommand{\mO}{\mathrm{O}}

\def\extd{\mathrm {d}}
\newcommand{\e}{\epsilon}

\newcommand\acts\triangleright

\newcounter{letter} \newcounter{numeral} \newcounter{Numeral}

\def\vphi{\varphi}

\def\e{\mbox{e}}
\def\extd{\mathrm {d}}

\newtheorem{theo}{Theorem}
\newtheorem{lemma}[theo]{Lemma}

\begin{document}



\title{$O(N)$ Random Tensor Models}

\author{Sylvain Carrozza}\email{sylvain.carrozza@labri.fr}\affiliation{Univ. Bordeaux, LaBRI, UMR 5800, 33400 Talence, France}
\author{Adrian Tanasa}
\email{adrian.tanasa@labri.fr}\affiliation{Univ. Bordeaux, LaBRI, UMR 5800, 33400 Talence, France}
\affiliation{H. Hulubei National Institute for Physics and Nuclear Engineering, P.O.B. MG-6, 077125 Magurele, Romania}
\affiliation{IUF, 1 rue Descartes, 75231 Paris Cedex 05}

\begin{abstract}
\bigskip

We define in this paper a class of three indices tensor models, endowed with $O(N)^{\otimes 3}$ invariance ($N$ being the size of the tensor). This allows to generate, via the usual QFT perturbative expansion, a class of Feynman tensor graphs which is strictly larger than the class of Feynman graphs of both the multi-orientable model (and hence of the colored model) and the $\U(N)$ invariant models. We first exhibit the existence of a large $N$ expansion for such a model with general interactions. We then focus on the quartic model and we identify the leading and next-to-leading order (NLO) graphs of the large $N$ expansion. Finally, we prove the existence of a critical regime and we compute the critical exponents, both at leading order and at NLO. This is achieved through the use of various analytic combinatorics techniques.
\end{abstract}
\maketitle

{\emph{MSC(2010) classification:} 83C27 ; 81T18 ; 05C30. 


\section{Introduction}

Random tensor models, seen as a natural generalization of the celebrated matrix models to dimension higher than two, imposed themselves over the last years as an emergent domain at the very frontier of theoretical physics, combinatorics and other domains of mathematics. 
The so-called \emph{colored models} \cite{razvan_colors, razvan_jimmy_rev} brought new combinatorial tools to tensor models, which tremendously improved our control over the definition and topological content of such models.
This allowed in particular to demonstrate the existence of a large $N$ expansion \cite{RazvanN, RazvanVincentN, razvan_complete}, which generalizes the genus expansion of matrix models in a non-trivial way. Though not a topological expansion -- as is to be expected in the absence of a simple classification of topological manifolds in dimension higher than two --, it opened the way to many developments and generalizations of matrix models results. In particular, it was proven that spherical triangulations of a particular type, called melonic triangulations, dominate the large $N$ expansion and lead to a continuum limit \cite{continuum_TM}. 
The random geometry defined by this melonic phase was later studied in more detail \cite{jr_branched}, and was confirmed to lie in the same universality class as the so-called branched polymer phase of dynamical triangulations \cite{book_ambjorn}. 

These simplicial colored models soon lead to a first generalization, the so-called \emph{uncolored} or \emph{invariant} tensor models \cite{uncolored}. 
%
%
Unlike simplicial models, which are constructed out of a single type of interaction (the $d$-simplex), 
the theory space of invariant tensor models is specified by a $\U(N)^{\otimes d}$ symmetry (where $d$ is the number of indices), which allows infinite sets of interactions. 
The colored simplicial structure of the original models is encoded in the non-trivial interior structure of the interaction vertices -- also called bubbles. Many interesting results have been gathered for both colored or uncolored tensor models, including the characterization of next-to-leading order Feynman graphs in the large $N$ expansion and the definition of a double-scaling limit \cite{ds_WDJ, ds_SRV, ds_VRJA}, or universality results \cite{universality} proven with constructive methods \cite{razvan_beyond, borel_TR}.

The \emph{multi-orientable} (MO) models, which have been defined in \cite{MO-original} (see also \cite{Tanasa-praa} for a short review), remained in the simplicial context but relaxed the coloring conditions of the original colored models. This type of models thus allowed to generate, {\it via} the usual QFT perturbative expansion, a class of Feynman graphs which is strictly larger than the class of graphs of the colored models. Moreover, the MO model were endowed with a large $N$ expansion mechanism \cite{DRT}. This expansion was analyzed, from a combinatorial point of view, first at leading \cite{DRT} and next-to-leading \cite{RT} orders and then at any order \cite{FT}. This allowed to implement the celebrated double-scaling mechanism \cite{GTY} - see also \cite{tanasa-last} for a recent review on these topics.


All of these purely combinatorial results are at the basis of more involved models, known as tensorial group field theories, which are field theories with the same bare combinatorial structure \cite{tensor_4d, tensor_3d, cor_u1, Dine_Fabien, cor_su2, VD}. 
Some of these models can be understood as toy-models for quantum gravity \cite{cor_su2, discrete_rg, 4-eps, VD}; they have been extensively studied through QFT 
renormalization techniques, see \cite{thesis} and references therein. 


\medskip

In order to unravel new aspects of random geometry in dimension higher than two, it appears crucial to us to aim at a step by step enlargement of the set of interactions and combinatorial structures allowed by the definition of our models. 

In this paper we perform such a step with the introduction of a new class of tensor models,
based on an $\mO(N)^{\otimes d}$ rather than a $\U(N)^{\otimes d}$ invariance. This leads to a larger class of allowed interactions, still labeled by colored graphs, but not necessarily bipartite. We focus on the $d=3$ situation, and show that these real tensor models contain the colored, uncolored and multi-orientable models as special cases. This is a very natural arena in which to further generalize tensor methods, which may provide a suitable enlarged theory space for tensorial group field theories.    

For the sake of completeness, let us mention that matrix models with $\mO(N)$ invariance have also been studied in papers such as \cite{zinn}, where tangle and link counting results have been proved.

\medskip

This paper is structured as follows. In the following section we introduce the general class of real tensor models, with arbitrary many (but a finite number of) interaction terms in the action. We prove the existence of a large $N$ expansion  which generalizes 
the large $N$ expansions of colored, uncolored and multi-orientable models. In section \ref{sec:quartic}, we restrict our attention to the most general quartic model (with invariance under permutation of the colors), parametrized by two independent coupling constants, and 
characterize the leading and next-to-leading orders in $1/N$. This allows us to prove, in section \ref{sec:critical}, the existence of a critical regime and compute its critical exponents, both at leading and next-to-leading orders. We rely to this effect on the analytic properties of the two-point function and on methods of analytic combinatorics -- which allow us to also derive some asymptotic estimates. 

\section{General model and large N expansion}


 In this section we define the tensor model we are interested in and we implement its large $N$ expansion. We consider a real tensor $T_{i_1 i_2 i_3}$ 
with three indices ranging from $1$ to $N \in \mathbb{N}^*$. The group $\mO(N)^{\otimes 3}$ acts on $T$ as
 \beq
 T_{i_1 i_2 i_3} \rightarrow O^{(1)}_{i_1 j_1} O^{(2)}_{i_2 j_2} O^{(3)}_{i_3 j_3} T_{j_1 j_2 j_3}
 \eeq
 where $O^{(k)}$ are three independent orthogonal matrices. 
 We further define the partition function
 \beq
 Z_N  := \int [\extd T] \, \e^{-N^{3/2} S_N (T)}
 \eeq
 where the action $S_N (T)$ is required to be invariant under $\mO(N)^{\otimes 3}$, and $[\extd T]$ is the product of the Lebesgue measures associated to the entries $T_{i_1 i_2 i_3}$. The action $S_N$ decomposes as a sum over tensor invariants $I_b$, weighted by coupling constants $t_b$:
 \beq
 \label{action}
 S_N (T) = \sum_{b \in \cB} t_b \, N^{- \rho(b) } \, I_b(T)\,, 
 \eeq
 where we have denoted by $\cB$ the set of $O(N)$ tensor invariants, and introduced a free scaling parameter $\rho(b)$ for each $b \in \cB$. These tensor invariants are in one-to-one correspondence with $3$-colored graphs\footnote{We refer in this paper to edge-coloring, with color labels $\ell \in \{ 1, 2, 3\}$.}
but \emph{not necessarily bipartite}. This is the only difference with respect to $\U(N)$ invariant models. The mapping between graphs and invariants is as follows: each node of the graph is associated to a tensor, and each line of color $\ell$ represents the contraction of the $\ell^{\mathrm{th}}$ indices of two tensors.

Let us recall that, in the tensor model framework, one calls these invariants {\it bubbles} and they are vertices of the QFT model to be considered (see for example \cite{uncolored} for details on this type of construction).


We furthermore restrict to connected invariants, that is to those invariants which are labeled by connected colored graphs.
In Figure \ref{examples_bubbles}, we have depicted the two- and four-point bubbles (up to a permutation of the color labels).
\begin{figure}[ht]
  \centering
 	\includegraphics[scale=0.8]{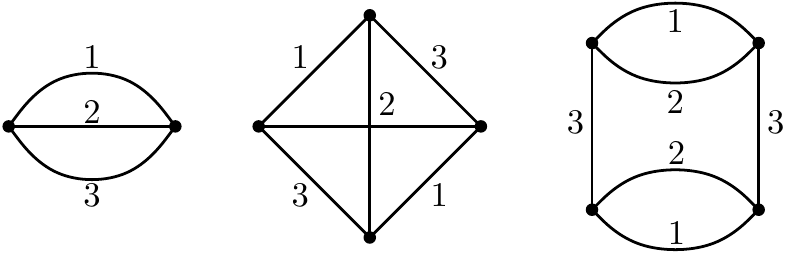}     
  \caption{The two- and four-valent bubbles/vertices of the model.}
  \label{examples_bubbles}
\end{figure}
The second bubble in Figure \ref{examples_bubbles} has the same combinatorial structure as the tetrahedron in spin foam \cite{review_marc, review_alejandro} and group field theory models \cite{freidel_gft, daniele_rev2011} for 3d quantum gravity. Note nonetheless that as a colored graph it encodes the topology of the projective plane \cite{Vince_2d}. The third bubble in Figure \ref{examples_bubbles} represents the so-called pillow term (see \cite{FL} and \cite{karim}), and like the first bubble of Figure \ref{examples_bubbles} is associated to a spherical triangulation.
The first and the third bubbles can be given a biparite structure such that they become vertices of the complex model. However,
the second bubble is not allowed by the $\U(N)^{\otimes 3}$ invariance; it is a new ingredient of the $\mO(N)^{\otimes 3}$ model. 
%

 The leftmost bubble of Figure \ref{examples_bubbles} -- call it $b_2$ -- encodes the quadratic part of the action. Since we will evaluate the model in perturbation around the Gaussian contribution, we fix the normalization of the invariant associated to $b_2$ to: $t_{b_2} = 1$ and $\rho( b_2) = 0$. The weight $N^{3/2}$ in front of the action is the natural one if one wants the second bubble of Figure \ref{examples_bubbles} to contribute at leading order (assuming that $\rho = 0$ for this bubble), as already observed in colored \cite{razvan_complete} and MO \cite{DRT} models. The function $\rho$ will be adjusted below after analysis of the amplitudes. Note that the usual uncolored models \cite{uncolored} come with a global weight $N^2$ instead, and individual weights $N^{-2 \omega(b)}$ (where $\omega (b)$ is the \emph{degree}  of the colored graph $b$) in front of the coupling constants $t_b$.

As usual, one splits the action between the quadratic part, which provides a notion of propagation, and the higher order terms which are Taylor expanded. 
As in ordinary uncolored tensor models \cite{uncolored}, the resulting Feynman amplitudes $\cA_\cG$ are labeled by graphs $\cG$, made out of bubble vertices and propagator lines, which are represented as dashed lines in our figures. Feynman diagrams are in particular $4$-colored graphs, a fourth color $0$ being attributed to the dashed lines. Three simple examples are given in Figure \ref{non-MO} and \ref{ex_ampl} (we will leave some of the color labels implicit in our pictures). 

As already mentioned in the introduction, the model we propose here generates a larger class of Feynman graphs than both the MO model (and hence than the colored one) and the $\U(N)$ invariant models. One has:
\begin{proposition}
\label{generalizare}
The sets of Feynman graphs generated by the MO action or a $\U(N)^{\otimes 3}$ invariant action are strict subsets of the set of Feynman graphs generated by 
the real action \eqref{action}.
\end{proposition}
\begin{proof}
One can easily check that the  interaction given by the second bubble in Figure \ref{examples_bubbles} already generates, by perturbative expansion, a strictly larger class of graphs than the class of MO tensor graphs. Indeed, MO graphs may be generated by the vertex of Figure \ref{MO_vertex}, with the additional condition that propagator lines must connect black nodes to white nodes. Hence all MO graphs may be realized in our $\mO(N)$ model, while for instance the Feynman graph of Figure \ref{non-MO} cannot be given the structure of a MO graph. Finally, the graphs generated by the $\U(N)^{\otimes 3}$ models coincides with the subclass of $\mO(N)^{\otimes 3}$ graphs which can be given a bipartite structure. It is a strict subclass since the graph of Figure \ref{non-MO} does not admit any bipartite coloring of its nodes.
\end{proof}

\begin{figure}[h]
  \centering
 	\includegraphics[scale=0.8]{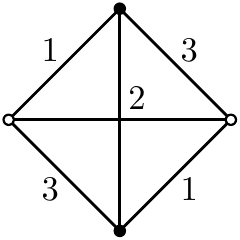}     
  \caption{The unique MO interaction vertex.}\label{MO_vertex}
\end{figure}

\begin{figure}[h]
  \centering
 	\includegraphics[scale=0.8]{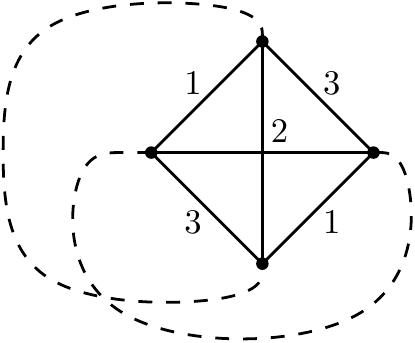}     
  \caption{A Feynman graph of the $\mO(N)^{\otimes 3}$ model which is neither MO nor bipartite.}\label{non-MO}
\end{figure}

One may perturbatively expand the partition function around the Gaussian, leading to:
\beq
Z_N = \sum_\cG \prod_{b \in \cB\vert N_b > 2} ( - t_b N^{- \rho(b)} )^{n_b(\cG)} \, \cA_\cG\,,
\eeq
where $\cA_\cG$ is the amplitude associated to a graph $\cG$, and $n_b (\cG)$ is the number of bubble vertices of type $b$ in $\cG$. A standard calculation (see e.g. \cite{RazvanN}) shows that the amplitude $\cA_\cG$ of a closed graph $\cG$ is
\beq
\cA_\cG \propto N^{3 - \omega (\cG)}\,,
\eeq
where we have defined the degree $\omega (\cG)$ as:
\beq
\omega (\cG) :=  3 + \frac{3}{2} L(\cG) - \sum_{b \vert N_b > 2} \left( \frac{3}{2} - \rho(b) \right) n_b (\cG) -  F(\cG)\,.
\eeq
In the definition above, we have denoted by $L$ and $F$ respectively the number of lines and faces\footnote{In this context, see e.g. \cite{uncolored}, a \emph{face} of color $\ell$ is a cycle of alternating color-$0$ lines and color-$\ell$ edges. We also define the \emph{length} of a face as its number of color-$0$ lines.} of $\cG$ (these are notations we will stick to throughout the paper).

This form of the Feynman amplitudes entails the existence of a $1/N$ expansion, provided that $\omega$ is bounded from below. We will show that this can be achieved. We will furthermore choose $\rho$ as small as possible, so that the class of leading order graphs in $N$ is as large as possible. 

In order to conveniently count the number of faces $F$, we introduce the notion of \emph{jacket}, which is defined similarly as in the rank-$3$ colored framework \cite{RazvanN, Ryan_jacket}. 
\begin{definition}
For any graph $\cG$ and any $\ell \in \{ 1, 2, 3 \}$, the \emph{jacket} $J_\ell (\cG)$ is the $3$-colored graph obtained from $\cG$ after deletion of all it color-$\ell$ edges.
\end{definition}
\noindent Equivalently, each $J_\ell$ is obtained by deletion of all the faces of color $\ell$, and hence represents a ribbon graph with faces of colors in $\{ 1, 2 , 3 \} \setminus \{ \ell \}$. A jacket therefore represents a closed and possibly non-orientable surface. Unlike in \cite{RazvanN, Ryan_jacket} the jacket of a connected graph is not necessarily connected -- see Fig. \ref{jacket} for examples of such jackets, associated to the tensor graphs of Figure \ref{ex_ampl}.

\begin{figure}[h]
  \centering
 	\includegraphics[scale=0.8]{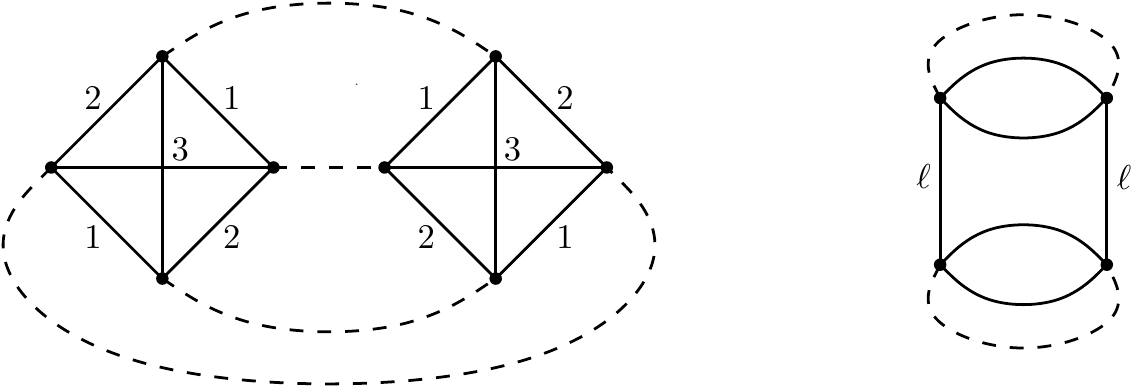}     
  \caption{Two vacuum (and melonic) graphs.}\label{ex_ampl}
\end{figure}

\begin{figure}[h]
  \centering
 	\includegraphics[scale=0.8]{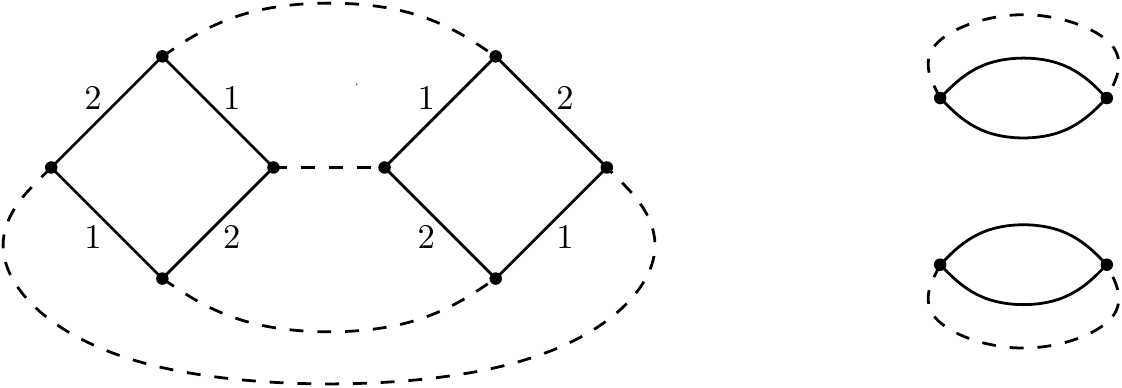}     
  \caption{Jackets associated to the tensor graphs of Figure \ref{ex_ampl}. On the right-hand-side, two connected components are generated by the removal of the lines of color $\ell$.}\label{jacket}
\end{figure}

Note that for graphs which are multi-orientable, the notion of jacket introduced above coincides with the one of multi-orientable jacket \cite{DRT}.

Faces can be counted as follows
\beq
2 F = \sum_{\ell=1}^3 f( J_\ell )
=  \sum_{\ell \, ; \, i} f( J_\ell^{(i)} ) \,,
\eeq
where $f ( J_\ell^{(i)} )$ is the number of faces of the $i^{\rm{th}}$ connected component of $J_\ell$. This number can be expressed in terms of the non-orientable genus $k$ as:
\beq
f( J_\ell^{(i)} ) = 2 - v ( J_\ell^{(i)} ) + e ( J_\ell^{(i)} ) - k ( J_\ell^{(i)} )\,,
\eeq
where $e$ (resp. $v$) is the number of edges (resp. vertices) of the jackets. 

To facilitate the analysis in the case of multiple connected components, let us introduce the quantity
\beq
\delta_\ell := \vert J_\ell \vert - 1\,,
\eeq
where $\vert J_\ell \vert$ denotes the number of connected components of $J_\ell$. Likewise, we define
\beq
\delta_\ell^{(b)} := C_b^{\ell} - 1\,,
\eeq
where $C_b^{\ell}$ is the number of connected components of the $2$-colored graph obtained from $b$ after deletion of all the color-$\ell$ lines. 
Remark that
\bes
\sum_{ i} e( J_\ell^{(i)} )  &=& L(\cG) = \sum_{b \vert N_b > 2} n_b \frac{N_b}{2}\,,  \\
\sum_{ i} v( J_\ell^{(i)} )  &=& \sum_{b \vert N_b > 2} n_b \left( 1 + \delta_\ell^{(b)} \right) \,,
\ees
where $N_b$ is the number of nodes of the bubble $b$ (i.e. the valency of the tensor interaction). In terms of these quantities, $\omega$ can be reexpressed in the form
\beq
\omega (\cG) = \frac{1}{2} \sum_{\ell \, ; \, i} k(J_{\ell}^{(i)}) + \sum_{b \vert N_b > 2} n_b \left( \rho(b) + \frac{1}{2} \sum_\ell \delta_\ell^{(b)} \right) - \sum_\ell \delta_\ell
\eeq

The key simple observation leading to the definition of a $1/N$ expansion is the following.
\begin{lemma}\label{lemma_ineq_rho} 
For any $\ell \in \{ 1,2,3 \}$, one has 
 \beq
 \sum_{b \vert N_b > 2} n_b \, \delta_\ell^{(b)} \geq \delta_\ell\,.
 \eeq
 Moreover, the number $n_b$ of vertices of each type being given, one can always construct a graph $\cG$ on these vertices which saturates the bound. 
\end{lemma}
\begin{proof}
By definition, each connected component of $J_\ell$ is supported on at least one connected component of the colored subgraph with colors $\{1,2,3\} \setminus \{ \ell \}$, hence the inequality. And one obtains an equality by pairing half-lines to half-lines in the same connected component of the color $\{1,2,3\} \setminus \{ \ell \}$ subgraph.   
\end{proof}

This allows us to conclude that the optimal choice for $\rho$ is:
\beq\label{def_rho}
\boxed{\rho(b) :=\frac{1}{2} \sum_\ell \delta_\ell^{(b)} }
\eeq
which we assume in the rest of the paper. Note that in terms of the number of colored faces\footnote{Where a \emph{colored face} of $b$ is a cycle of alternating colors $i$ and $j \in \{1,2,3\}$ in $b$.} $F_b$ of $b$, this definition coincides with: 
\beq
\rho(b) = \frac{F_b - 3}{2}\,.
\eeq
\begin{proposition}
Assuming a scaling of the coupling constants as defined by (\ref{def_rho}) and (\ref{action}), the amplitude of a graph $\cA_\cG$ is proportional to $N^{3 - \omega(\cG)}$, where the degree can be expressed as:
\beq\label{omega_j}
\omega = \frac{1}{2} \sum_{\ell \, ; \, i} k(J_{\ell}^{(i)}) + \sum_{b \vert N_b > 2} n_b  \sum_\ell \delta_\ell^{(b)} - \sum_\ell \delta_\ell\,. 
\eeq
Furthermore, $\omega(\cG) \in \frac{\mathbb{N}}{2}$.
\end{proposition}
\begin{proof}
The expression of $\omega$ is a simple consequence of the relations defining $\rho$ and $\omega$ and the other combinatorial quantities. Equation (\ref{omega_j}), together with the fact that the demigenus $k$ is an integer\footnote{The demigenus (or half-genus) $k$ is usually invoked to label the topology of non-orientable surfaces, while for orientable surfaces, the genus $g = \frac{k}{2}$ is in general preferred. We will work with the demigenus throughout the paper, irrespectively of the orientability of the surface. This means in particular that $k$ will be even in the case of an orientable surface, and odd in the case of a non-orientable one.}, implies that $\omega \in \frac{\mathbb{Z}}{2}$. Furthermore, Lemma \ref{lemma_ineq_rho} and the positivity of $k$ imply that $\omega$ is itself positive.
\end{proof}

Remark that the choice of weight (\ref{def_rho}) ensures that:
\begin{equation}\label{morphism_rho}
\rho( b_1 \sharp b_2 ) = \rho( b_1 ) + \rho( b_2 ), \quad \forall b_1 , b_2 \in \cB\,, 
\end{equation}
where we loosely denote by $b_1 \sharp b_2$ any connected sum of $b_1$ and $b_2$, that is any bubble obtained by first connecting $b_1$ and $b_2$ with a propagator and then contracting this propagator. Such a contraction is called \emph{$1$-dipole} contraction in the literature. 
{\it A posteriori}, an independent motivation for our choice of scaling function $\rho$ is that it renders the degree $\omega$ invariant under these $1$-dipole contractions.

\begin{proposition} The leading order graphs are characterized by:
\bes
\forall (\ell , i ),&& \qquad k(J_{\ell}^{(i)}) = 0\,, \label{cond_k} \\
\forall \ell, && \qquad \delta_\ell = \sum_{b \vert N_b > 2} n_b \, \delta_\ell^{(b)}. \label{cond_delta}
\ees
\end{proposition}
\begin{proof}
The two terms appearing in the expression of the degree \eqref{omega_j} are positive or $0$. Hence they must both be $0$ when $\omega =0$, leading to equations \eqref{cond_k} and \eqref{cond_delta}. Furthermore, it is easy to check that the class of degree-$0$ graphs is non-empty: for instance, the two graphs shown in Figure \ref{ex_ampl} have $0$ degree. 
\end{proof}
At least melonic graphs in the sense of both colored and uncolored models are dominant. 
We will refer to these types of melons as respectively of type I and II. It is also easy to generate leading order graphs which are genuinely new, for example by contraction of an arbitrary number of tree lines in a $\vphi^4$ leading order graph of type I (this will create non-bipartite bubbles of valency higher than $4$ without changing the degree).

\bigskip

As announced in the introduction 
classifying leading order graphs 
lies outside the purpose of this paper.
In the remaining sections, we focus instead on the quartic real model.

\section{Quartic model, large N expansion}\label{sec:quartic}


The action of the quartic model writes:
\beq
S_N = \frac{1}{2} I_{\vcenter{\hbox{\includegraphics[scale=0.15]{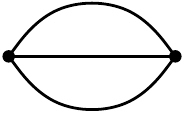}}}} + \frac{ \lambda_{1}}{4} \, I_{1} + \frac{\lambda_{2}}{12 \sqrt{N}} \, I_{2} \,,
\eeq
where $\lambda_1 , \lambda_2 \in \mathbb{R}$ and
\beq
I_1 := I_{\vcenter{\hbox{\includegraphics[scale=0.15]{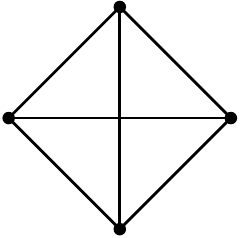}}}}\;,  \qquad I_2 := \sum_{\ell = 1}^3 I_{\vcenter{\hbox{\includegraphics[scale=0.15]{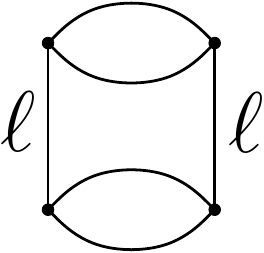}}}}\,.
\eeq
The interactions we allow are therefore: the tetrahedral interaction term, for which $\rho = 0$; and the three pillow invariants which have $\rho = \frac{1}{2}$, hence the scaling of $N^{-1/2}$ in front of $I_2$.
We will denote by $n_1$ the number of tetrahedral interactions in a given graph, and by $n_2$ the number of pillow vertices. Note that we have introduced normalization factors which take into account the symmetries of the respective interactions (this will of course facilitate the enumeration of graphs in the sequel).
We took into account:
\begin{itemize}
\item the number of automorphisms\footnote{That is, graph automorphisms which also preserve the coloring.} of each bubble (four in both cases), and 
\item an additional $1/3$ factor for the pillow interactions (this 
comes from the fact that we also require invariance under permutation of the color labels, and therefore use a single coupling constant for the three pillow interactions).
\end{itemize}
 
The invariant $I_2$ is an explicitly positive interaction, but $I_1$ is not. It is therefore likely that $S_N$ itself is unbounded from below, making the definition of the path integral questionable\footnote{We refer to \cite{freidel_louapre} in which this question is explored at length, though in a slightly different context.}. Our calculations should be interpreted as formal manipulations of power series in $\lambda_1$ and $\lambda_2$, and $N$. This being said, we will prove that the leading order graphs of this model are exactly the melonic graphs of colored and multi-orientable tensor models. The method of proof that we will use is itself a generalization of these cases. 

\subsection{Large N expansion; leading order}



In this model, we can define two types of melonic moves, called type I (resp. type II) contractions or insertions, which are nothing but the melonic moves already relevant in multi-orientable (resp. invariant) tensor models. See Figure \ref{melonic_moves}. 
This moves allow us to formalize the notion of melonic graph in this context.
\begin{definition} The family of vacuum \emph{melonic graphs} is the set of graphs generated by the two graphs shown in Figure \ref{ex_ampl}, and the melonic insertion operations of type I and II (see Figure \ref{melonic_moves}).
\end{definition}

It is also easily seen that the melonic moves conserve the degree.  

\begin{figure}[htp]
  \centering
  \subfloat[Melonic move of type I.]{\label{type1}\includegraphics[scale=0.8]{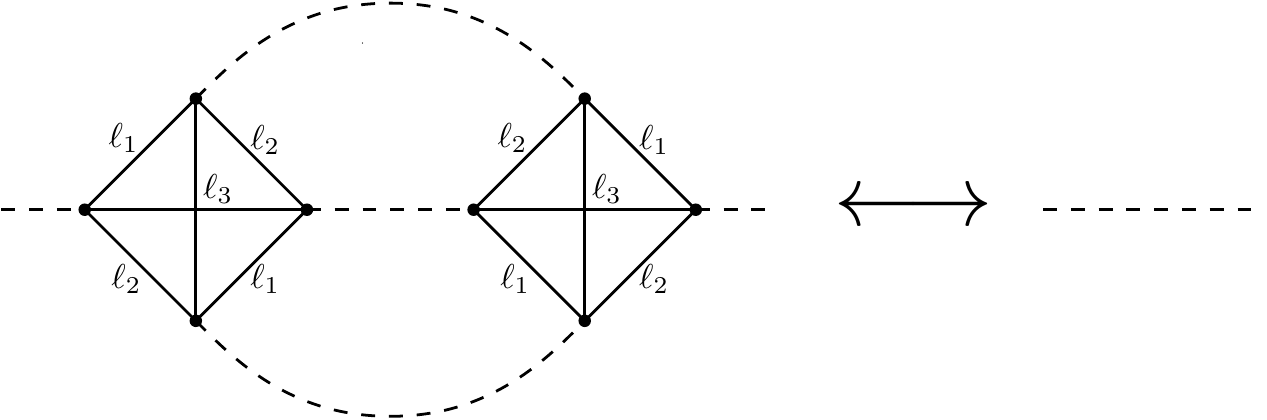}}
  \hspace{5pt}
  \subfloat[Melonic move of type II.]{\label{type2}\includegraphics[scale=0.8]{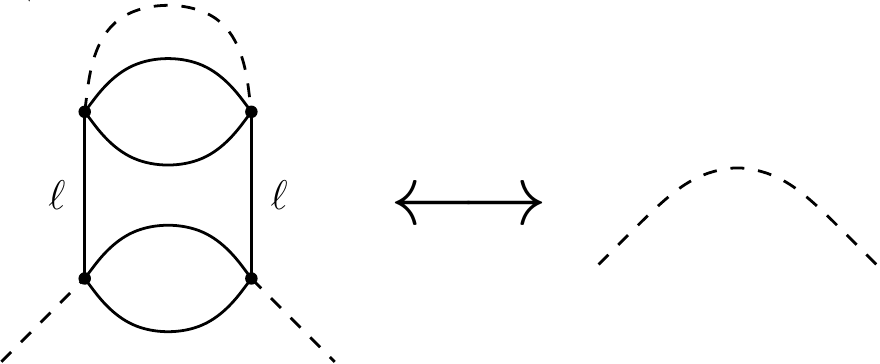}}
  \caption{Melonic moves in the quartic model: going from left to right is called a melonic contraction, while the reverse operation is a melonic insertion.}
  \label{melonic_moves}
\end{figure}

\begin{lemma}
 The degree $\omega$ is invariant under the melonic moves of type I and II.
\end{lemma}
\begin{proof}
 A melonic move of type I changes $L$ to $L-4$, $n_{1}$ to $n_{1} -2$, $F$ to $F-3$, and does not change $n_{2}$. Hence the new degree is
 \beq
 \omega - \frac{3}{2} 4 + 2 \left( \frac{3}{2} - 0 \right) + 3 = \omega\,.
 \eeq
 Similarly, a melonic move of type II changes $L$ to $L-2$, $n_{2}$ to $n_{2} -1$, $F$ to $F-2$, and does not change $n_{1}$. Hence the new degree is:
 \beq
 \omega - \frac{3}{2} 2 + \left( \frac{3}{2} - \frac{1}{2} \right) + 2 = \omega\,.
 \eeq
\end{proof}
Hence we have already proven that vacuum melonic graphs have degree $0$ and are therefore leading order. We will now prove that they are the only ones.

\

Let us first gather some facts about graphs $\cG$ with $n_{2}(\cG) = 0$. 
\begin{lemma}\label{lemma_length} Let $\cG$ be a vacuum graph such that $n_{2}(\cG) = 0$.
\begin{enumerate}[(i)]
 \item If $\cG$ has a face\footnote{We recall that the length of a face is defined as its number of color-$0$ lines.} of length $1$ then $\omega(\cG) \geq \frac{1}{2}$. 
 \item If $\cG$ has a face of length $3$, then $\omega(\cG) \geq \frac{1}{2}$.
\end{enumerate}
\end{lemma}
\begin{proof}
 (i) $\cG$ is either: a single-vertex graph with two tadpole lines i.e. the \emph{infinity graph} of Figure \ref{infty_graph}, in which case we explicitly check that $\omega(\cG) = \frac{1}{2}$; or it contains a non-trivial $2$-point graph with a single vertex (Figure \ref{2point_41}). The contraction of this $2$-point graph changes $L$ to $L-2$, $n_{1}$ to $n_{1} - 1$ and $F$ to $F-1$. Hence this yields a new graph $\cG'$ with degree $\omega(\cG') = \omega(\cG) - \frac{1}{2}$. The positivity of $\omega$ immediately implies that $\omega(\cG) \geq \frac{1}{2}$.
 
 (ii) The only way for $\cG$ to have a face of length $3$ without having also a face of length $1$ is as pictured in Figure \ref{proof_lemma_faces}. The jacket $J_{\ell_1}$ of this graph contains a ribbon with three twists and is therefore not orientable. This implies that at least one jacket, being not orientable, has a half-integer genus, and hence $\omega (\cG) \geq \frac{1}{2}$.
\end{proof}

\begin{figure}[b]
  \centering
  \subfloat[Infinity graph.]{\label{infty_graph}\includegraphics[scale=0.8]{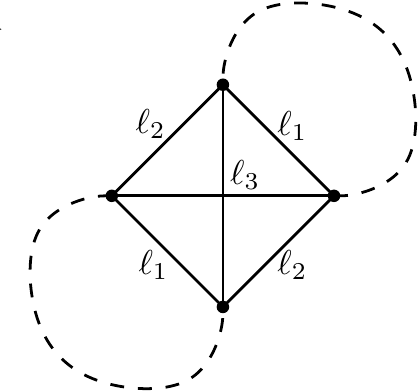}}
  \hspace{5pt}
  \subfloat[A $2$-point graph with a face of length $1$.]{\label{2point_41}\includegraphics[scale=0.8]{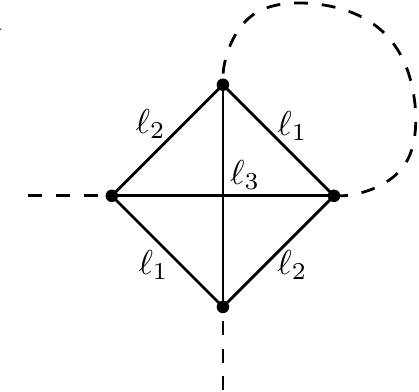}}
  \hspace{5pt}
  \subfloat[Face of length $3$ whose jacket yields a twisted ribbon graph.]{\label{proof_lemma_faces}\includegraphics[scale=0.8]{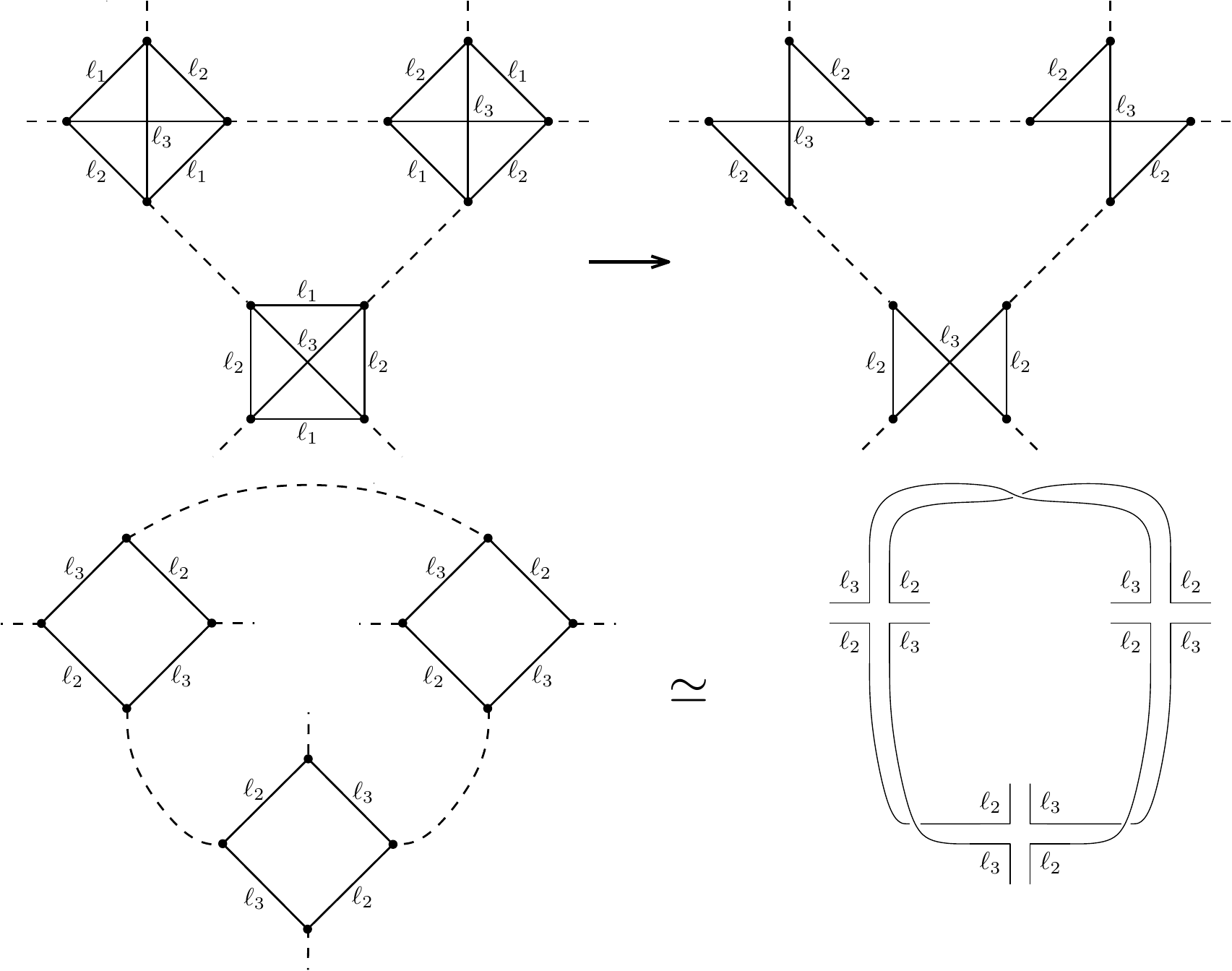}}\caption{Proof of Lemma \ref{lemma_length}.}
  \label{proof_lemma_length}
\end{figure}

\begin{lemma}\label{l_f2} Let $\cG$ be a vacuum graph such that $n_{2}(\cG) = 0$.
If $\omega(\cG) = 0$, then $\cG$ has at least $6$ faces of length $2$.
\end{lemma}
\begin{proof}
 Let $F_p$ be the number of faces of length $p$ in $\cG$, and $V$ the number of vertices. We have:
 \beq
 \sum_p F_p = F  = 3 + \frac{3}{2} L - \frac{3}{2} V \,, \qquad \sum_p p F_p = 3 L\,,
 \eeq
 where the condition $\omega(\cG) = 0$ was used in the first equation. From Lemma \ref{lemma_length}, we know that $F_1 = F_3 = 0$. Subtracting the second equation to four times the first therefore yields:
 \beq
 2 F_2 + \sum_{p \geq 4} (4 - p) F_p = 12 + 3 L - 6 V\,,
 \eeq
 and since moreover $L = 2 V$ we conclude that
 \beq
 F_2 = 6 + \frac{1}{2} \sum_{p \geq 4} (p - 4) F_p \geq 6\,.
 \eeq
\end{proof}

We also remark the following:
\begin{lemma}\label{l_n42}
 Let $\cG$ be a vacuum leading order graph. If $n_{2}(\cG) \neq 0$, then $\cG$ is of the form shown in Figure \ref{split}.
\end{lemma}
\begin{proof}
 One necessarily has a bubble $b$ and a color $\ell$ such that $\delta_\ell^{(b)} \geq 1$. Hence since $\cG$ is leading order, $\delta_\ell = \underset{b \vert N_b > 2}{\sum} n_b \, \delta_\ell^{(b)} \geq 1$. A splitting of a jacket of color $\ell$ in two connected components can only happen at a pillow vertex with color $\ell$, hence $\cG$ is of the form shown in Figure \ref{split}.  
\end{proof}

\begin{figure}[ht]
  \centering
 	\includegraphics[scale=0.8]{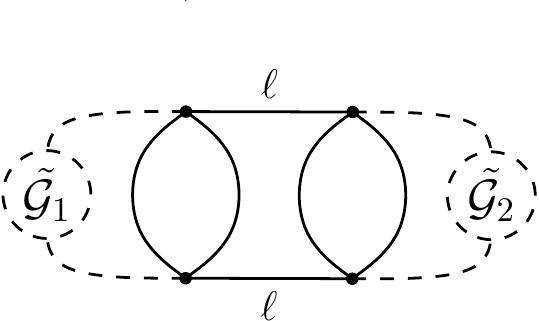}     
  \caption{Structure of a leading order graph containing a pillow bubble.}\label{split}
\end{figure}

\

Let us define a $2$-point leading order graph as a $2$-point graph which closes into a vacuum leading order graph.

\begin{proposition}\label{propo_l}
 Any $2$-point leading order graph contains a type I or type II elementary melon.
\end{proposition}
\begin{proof}
 By induction on $p = \lfloor \frac{n_{1}}{2} + n_{2} \rfloor$.
 
 $p=1$: The only two possible graphs are shown on the left side of Figures \ref{type1} and \ref{type2}. The fact that there are no other possible graphs with $n_{1} = 2$ and $n_{2} = 0$ is easily deduced from the fact that all faces must have length $2$ in this case (Lemma \ref{l_f2}).
 
 $p \geq 2$: If $n_{2} \neq 0$, choose a pillow vertex in $\cG$. If it does not directly provide an elementary melon of type II, then Lemma \ref{l_n42} ensures that the graph is of the form shown in Figure \ref{split}, where the condition $\omega = 0$ imposes that $\tilde{\cG_1} \neq \emptyset$ is a $2$-point leading-order graph with strictly lower $p$. By the induction hypothesis it therefore itself contains an elementary melon of type I or II. If $n_{2} = 0$, by Lemma \ref{l_f2} the graph is of the form shown in Figure \ref{form_lo}. If this does not already provide an elementary melon of type I, then we can perform the move shown in Figure \ref{move_4142}, which as is easily proved conserves the face structure and the degree. Now $n_{2} \neq 0$ with $p$ unchanged, and we can run the previous argument again to conclude. 
\end{proof}

\begin{figure}[htp]
  \centering
  \subfloat[Structure of a leading order graph with $n_{4,2}=0$.]{\label{form_lo}\includegraphics[scale=0.8]{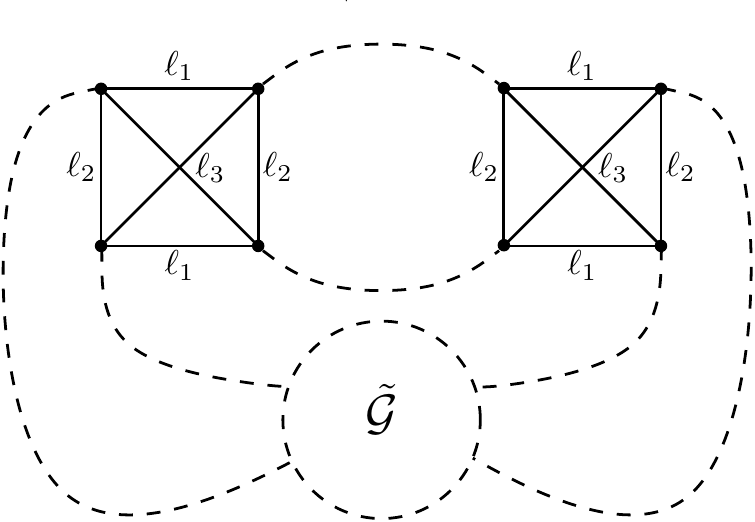}}
  \hspace{5pt}
  \subfloat[A move leaving the degree invariant.]{\label{move_4142}\includegraphics[scale=0.8]{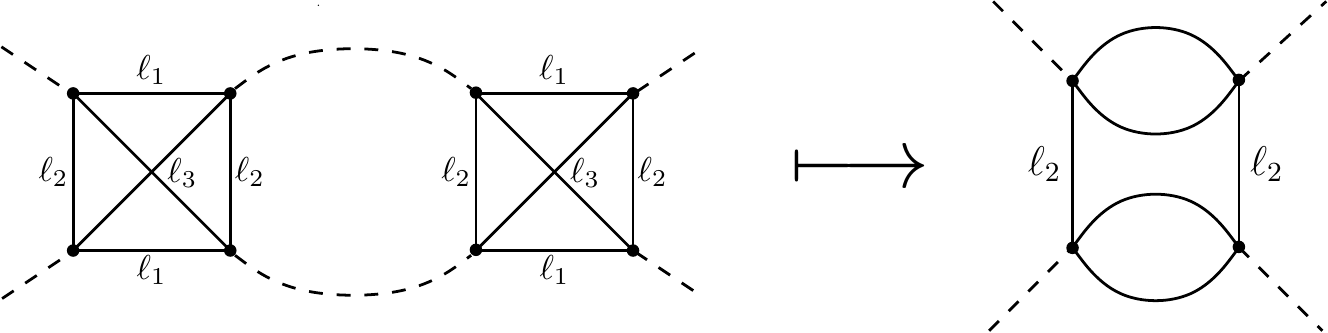}}
  \caption{Proof of Proposition \ref{propo_l}.}
  \label{proof_propo}
\end{figure}

\

This concludes our characterization of the leading order sector of the quartic model. 
\begin{proposition}
The vacuum leading order graphs of the quartic model are the vacuum melonic graphs.
\end{proposition}
\begin{proof}
By Proposition \ref{propo_l}, any leading order graph can be reduced to one of the two graphs of Figure \ref{ex_ampl} by successive contractions of melons. Hence such a graph is melonic.
\end{proof}

\subsection{Next-to-leading order}

We now focus on the next-to-leading order (NLO) graphs, which are characterized by $\omega = \frac{1}{2}$. From Eq. (\ref{omega_j}), we infer that graphs with degree $\omega = \frac{1}{2}$ must have a single non-trivial jacket of demigenus $k=1$ and should also verify condition (\ref{cond_delta}). Therefore Lemma \ref{l_n42} also holds for next-to-leading order vacuum graphs. 

One can check that there are exactly 3 single-vertex graphs with $\omega=\frac{1}{2}$: there are the so-called 'infinity' graphs represented in Figure \ref{infty_graph} (there is one such graph for each value of $\ell_1$). We readily obtain an infinite family of NLO graphs by insertion of non-trivial melonic $2$-point graphs, which as we have seen do not change the degree. In order to determine whether this family exhausts the set of NLO graphs, we follow \cite{RT} and introduce the notion of \emph{core graph}.

\begin{definition}
A \emph{core graph} is a vacuum graph with no melonic $2$-point subgraph.
\end{definition}

The question now is whether there exists more NLO core graphs than the three infinity graphs. Let us first prove 
the following lemma:

\begin{lemma}\label{NLO_faces}
Let $\cG$ be a NLO core graph. Then:
\begin{enumerate}[(i)]
\item $n_{2} (\cG) = 0$;
\item if $\cG$ is made of more than 1 vertex, then all its faces have length higher or equal to 3. Moreover, all the faces of length 3 must have the same color as the non-planar jacket of $\cG$.
\end{enumerate}
\end{lemma}
\begin{proof}
(i) Assuming  $n_{2} (\cG) \neq 0$, condition (\ref{cond_delta}) imposes that $\cG$ has the structure shown in Figure \ref{split}. Since $\cG$ is a core graph, $\omega(\tilde{\cG_1}) \geq \frac{1}{2}$ and $\omega(\tilde{\cG_2}) \geq \frac{1}{2}$. But one can check (from the definition of $\omega$) that $\omega(\cG) = \omega(\cG_1) + \omega(\cG_2)$, which yields $\omega(\cG) \geq 1$ in contradiction with the NLO character of $\cG$.

(ii) Since there is no pillow vertex in $\cG$, a face of length $1$ has to be of the form shown in Figure \ref{2point_41}. Let us call $\tilde{\cG}$ the non-trivial $2$-point subgraph which closes this figure. By hypothesis, $\omega(\tilde{\cG}) \geq \frac{1}{2}$. One can also verify that $\omega(\cG) = \omega(\tilde{\cG}) + \frac{1}{2}$, which is inconsistent with $\cG$ being NLO. Therefore $\cG$ has no face of length $1$. 
If $\cG$ had a face of length $2$, it would have the structure of Figure \ref{form_lo}. Performing the move of Figure \ref{move_4142}, which does not change the degree, would lead to a NLO graph with $n_{2} \neq 0$. This graph would therefore have to contain a $2$-point melonic subgraph, and so would $\cG$, but this cannot be since $\cG$ is a core graph.
Finally, just like before, the only way for $\cG$ to have a face of length $3$ without having also a face of length $1$ is as pictured in Figure \ref{proof_lemma_faces}. The existence of such a face of color $\ell_1$ imposes that the jacket of color $\ell_1$ is non-orientable. Since $\cG$ has exactly one non-planar jacket, all the faces of length $3$ must have color $\ell_1$. 
\end{proof}

This is sufficient to show that the infinity graphs are the only NLO core graphs in this model.

\begin{proposition}
The only NLO core graphs of the quartic model are the infinity graphs of Figure \ref{infty_graph}.
\end{proposition}
\begin{proof}
Let us assume that $\cG$ is a NLO core graph with more than 1 vertex and look for a contradiction. By Lemma \ref{NLO_faces} (i), $\cG$ cannot have any pillow vertex, and therefore its degree can be expressed as $\omega = 3 + \frac{3}{2} V - F$. We know that $\cG$ has three jackets, two of which (say $J_1$ and $J_2$) are planar and the last one being of demigenus $k = 1$. We also know by Lemma \ref{NLO_faces} (ii) that the faces of length $3$ have the same color as the non-planar jacket, namely the color 3. The total number of faces $F$ can be split into the $f(J_1 )$ faces of $J_1$ (which are of color $2$ and $3$) and the $F_1$ faces of color $1$: $F = f(J_1) + F_1$. $J_1$ being planar, Euler's relation implies $f(J_1) = v(J_1) + 2 = V + 2$, hence:
\begin{equation}
\omega(\cG) = 3 + \frac{3}{2} V - ( V + 2 + F_1 ) = 1 + \frac{1}{2} V - F_1\,.
\end{equation}
The $F_1$ faces of color $1$ have length higher or equal to 4, and each line of $\cG$ contains exactly one ribbon line of color 1, therefore:
\begin{equation}
2 V = L = \sum_{p \geq 4} p F_1^{(p)} \geq 4 F_1\,,
\end{equation}
where $F_1^{(p)}$ is the number of faces of length $p$ and color $1$. One concludes that
\begin{equation}
\omega(\cG) \geq 1 + \frac{1}{2} \times 2 F_1 - F_1 = 1\,,
\end{equation}
which contradicts the fact that $\cG$ is an NLO Feynman graph.
\end{proof}

\section{General quartic model: critical behaviour}\label{sec:critical}

The leading-order connected and one particle irreducible Green functions are proportional to a product of Kronecker delta functions. Let us call $G_{LO} (g, \mu)$ (resp. $\Sigma_0 (g, \mu)$) the proportionality factor of the connected (resp. one particle irreducible) Green function, in terms of the following parametrization of the coupling constants:
\begin{equation}
\label{parametrizare}
g := {\lambda_{1}}^2\,, \qquad \mu := - \frac{\lambda_{2}}{{\lambda_{1}}^2}\,. 
\end{equation}
In this way, the variable $g$ will allow to keep track of the total number of elementary melonic insertions, while $\mu$ will count the number of elementary melonic insertions of type II.

\subsection{Explicit counting of melonic graphs}\label{sec:counting}

The melonic graphs of our model can be represented by unlabelled colored trees. More precisely, with the weight we have chosen in the action, $G_{LO}$ writes:
\begin{equation}\label{def_G0}
G_{LO} (g, \mu) = \sum_{p, q \in \mathbb{N}} \, C_{p,q} \, g^{p+q} \mu^q\,, 
\end{equation}
where $C_{p,q}$ is the number of melonic $2$-point graphs with $p$ type I melons and $q$ type II melons, up to local color permutations of type II vertices\footnote{This is because and the reason why we have divided the corresponding interaction term by an extra factor $3$ in the action.}. Such melonic graphs can be recursively constructed by successive insertions of type I and type II elementary melons. They can therefore be represented by abstract trees with edge color labels, which record the location and type of the successive melonic insertions. One can for instance adopt the convention of Fig. \ref{tree_mapping}. Each equivalent class of melonic $2$-point graphs (up to color relabelling at the type II vertices) is represented by a rooted colored tree, an \emph{admissible coloring} of the edges of a tree being as follows: the $4$ lines outgoing (the notion of outgoing being defined with respect to the root) from a coordination $5$ vertex are labeled with integers from $0$ to $4$, the $2$ lines outgoing from a coordination $3$ vertex have labels $0$ and $1$, and finally the unique edge outgoing from the root vertex has color $0$.
In this manner, $C_{p,q}$ counts the number of \emph{colored} rooted trees with $p$ vertices of coordination $5$, $q$ vertices of coordination $3$, and $3p + q + 2$ leaves (including the root vertex). Note also that a coloring is equivalent to a choice of local orientation around each vertex of the tree, hence $C_{p,q}$ equivalently counts the number of \emph{binary--quaternary plane trees} with $p$ quaternary and $q$ binary vertices. See Figure \ref{ex_trees} for an example of a tree representation of a melonic $2$-point function.

\begin{figure}[ht]
  \centering
 	\includegraphics[scale=0.6]{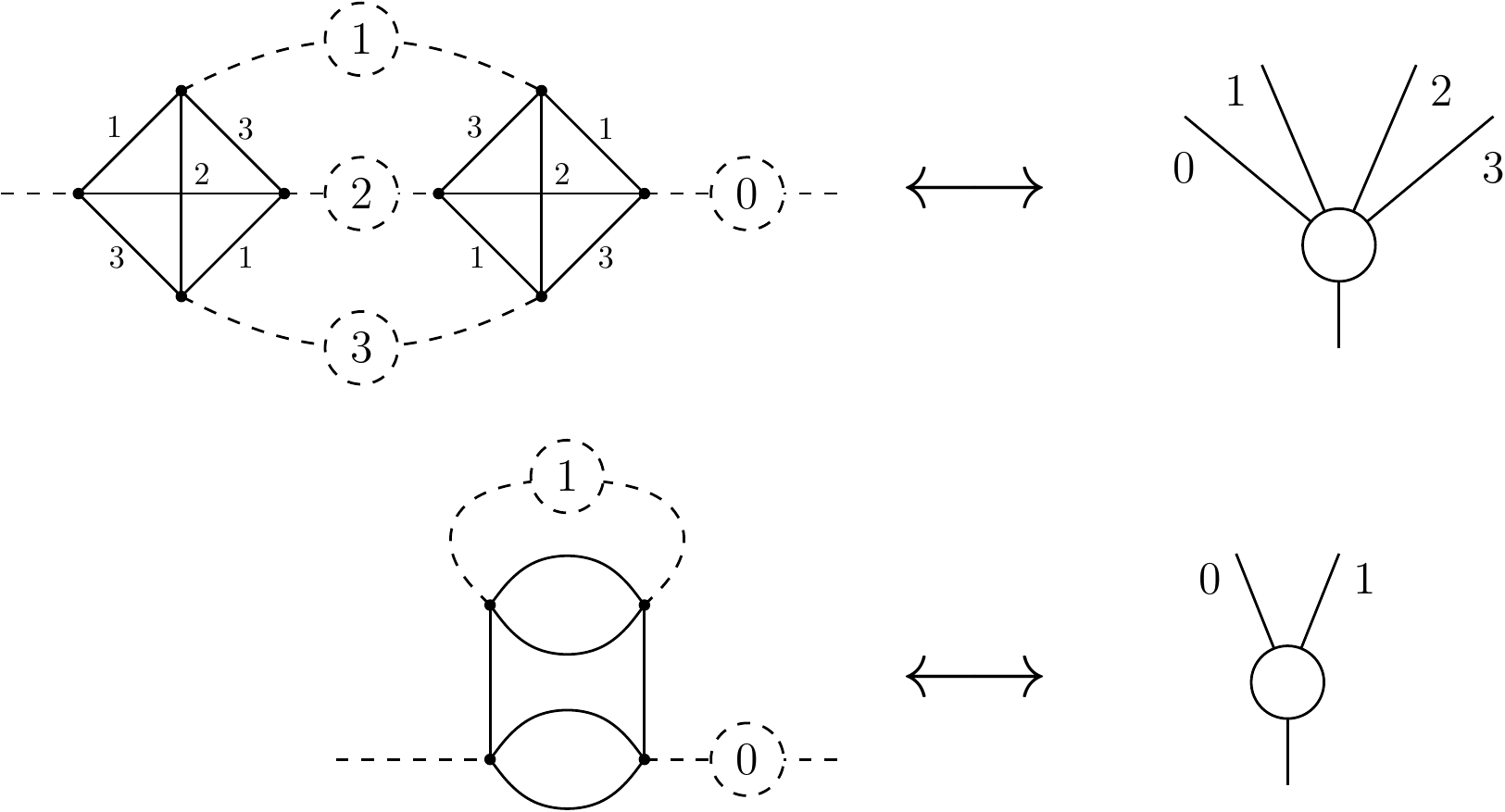}     
  \caption{Correspondence between elementary melons and plane tree vertices. Each labeled circular dashed region on the left-hand-side corresponds to the open leg with the same label on the right-hand-side.}\label{tree_mapping}
\end{figure}

\begin{figure}[ht]
  \centering
 	\includegraphics[scale=0.6]{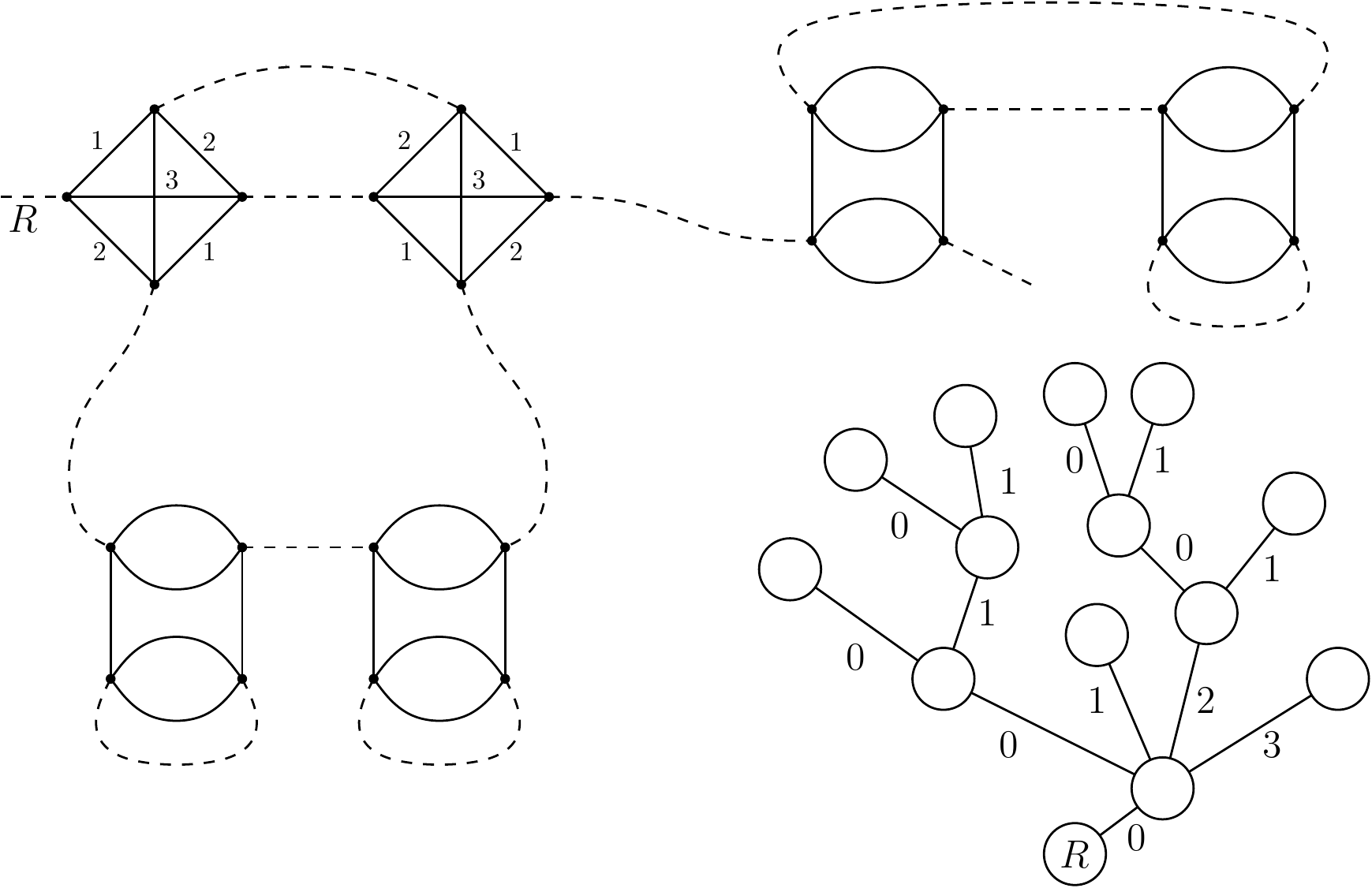}     
  \caption{A $2$-point melonic graph with root external leg $R$ and its tree representation.}\label{ex_trees}
\end{figure}

\

By Cayley's theorem, the number of labelled trees with $p$ vertices of valency $5$, $q$ vertices of valency $3$, $3p + q + 2$ leaves, and therefore a total number of $4p + 2q + 2$ vertices, is
\begin{equation}
\frac{(4p + 2q)!}{(4!)^p \, 2^q}\,.
\end{equation}
The outgoing edges of each coordination $5$ vertex admit $4!$ distinct colorings or orientations, while the outgoing edges of coordination $3$ vertices admit only $2$ distinct colorings. This gives a total multiplicative factor of $(4!)^p \, 2^q$. Since our trees are furthermore unlabelled and rooted, one should divide by the number of possible permutations of the vertices except for the root: namely, one should divide by $p!$ (for the coordination 5 vertices), $q!$ (for the coordination $3$ vertices), and  $(3p + q + 1)!$ (for all the coordination $1$ vertices but the root). This yields:
\begin{equation}\label{enum_trees}
C_{p,q} = \frac{[4 p + 2 q]!}{p! \, q! \, (3 p+ q + 1)! }\,.
\end{equation}

\

\noindent\emph{Remark.} Tree structures can naturally be enumerated by means of quantum field theory techniques \cite{malek, continuum_TM, andrea}. In the present situation one can for instance define the partition function
\begin{equation}
Z( g, \mu, J) = \int \extd \mu_C \, \exp\left( \vphi J + g \, \vphi \psi^4 - \mu \, g \, \vphi \psi^2 \right)\,,
\end{equation}
where the covariance $C$ of the Gaussian measure $\mu_C$ is defined by: 
\beq
C(\vphi, \vphi) =  C( \psi , \psi ) = 0 \,, \qquad C( \vphi, \psi )  = 1\,.
\eeq
The inspection of the perturbative expansion of this auxiliary field theory shows that the melonic Green function $G_{LO}$ evaluates as the connected expectation value:
\begin{equation}
G_{LO} (g , \mu ) = \frac{\langle \psi_0 \rangle_c (g, \mu , J=1)}{Z( g, \mu , J = 1 ) }\,.
\end{equation}
One can easily compute $Z(g , \mu , J)$ and $\langle \psi_0 \rangle_c (g, \mu , J)$ perturbatively in $\mu$ and $J$, and hence reduce formula (\ref{enum_trees}) to the computation of a product of two formal power series. Given the simplicity of the previous proof we will not give more details here, the interested reader is referred to \cite{continuum_TM} for a similar calculation. 

\

In order to obtain a first crude understanding of the divergence structure of $G_{LO}$, one may resort to the following asymptotics of the coefficients $C_{p,q}$.
\begin{proposition} The coefficients $C_{p,q}$ have the following asymptotic behaviour:
\begin{enumerate}[(i)]
\item For any $q_0 \in \mathbb{N}$:
\beq
C_{p, q_0 } \underset{p \to + \infty }{\sim} \frac{1}{3} \sqrt{\frac{2}{3 \pi}} \frac{1}{q_0 !} \left( \frac{16}{3} \right)^{q_0 } \, p^{ q_0 - 3/2} \left( \frac{4^4}{3^3} \right)^p
\eeq
\item For any $p_0 \in \mathbb{N}$:
\beq
C_{ p_0 , q } \underset{q \to + \infty }{\sim} \frac{1}{\sqrt{\pi}} \frac{1}{p_0 !} 16^{p_0 } \, q^{ p_0 - 3/2} \, 4^q
\eeq
\end{enumerate}
\end{proposition}
\begin{proof}
These expressions are direct consequences of Stirling's formula.
\end{proof}

As a consequence of Fubini's theorem, if the right-hand-side of equation (\ref{def_G0}) is absolutely convergent then the partial sums over $p$ and $q$ respectively are absolutely convergent and therefore:
\beq
\vert g \vert < \frac{3^3}{4^4} \qquad \mathrm{and} \qquad \vert \mu g \vert < \frac{1}{4}\,.
\eeq
This already proves the existence of a critical regime for $G_{LO}$, but this is not enough to locate the singularities or compute the critical exponents. Indeed, $G_{LO}$ can be rewritten as:
\begin{equation}
G_{LO} (g, \mu) = \sum_{n \in \mathbb{N}} \alpha_n (\mu) \, g^n \,,  
\end{equation}
with
\begin{eqnarray} 
\alpha_n (\mu) &=& 
\sum_{q=0}^n \frac{\mu^q }{ 3 n - 2 q + 1} {{n }\choose{q}} {{4n - 2 q}\choose{n }} 
\end{eqnarray}
and it is the asymptotic behaviour of $\alpha_n (\mu)$ which determines the critical behaviour of $G_{LO}$. Rather than directly evaluating this complicated sum, we will adopt an alternative strategy based on the analysis of the analytic properties of an algebraic equation for $G_{LO}$. Interestingly, not only this method does not require any evaluation of $\alpha_n (\mu)$ (or any prior knowledge of the coefficients $C_{p,q}$), it actually yields the asymptotic formula for $\alpha_n (\mu)$ we are missing as a bonus. 


\subsection{Diagrammatic equations, leading and next-to-leading order}

We know from the standard relation between the connected and one particle irreducible Green functions that:
\begin{equation}
G_{LO} = \frac{1}{1 - \Sigma_0}\,.
\end{equation}
From the structure of melonic graphs, we can furthermore infer the relation depicted in Figure \ref{schwinger_dyson_lo}. Carefully taking combinatorial factors and signs into account, this leads to: 
\begin{equation}
\Sigma_0 = {\lambda_{1}}^2 \, {G_{LO}}^3 - \lambda_{2} \, G_{LO} = g \, {G_{LO}}^3 + g \, \mu \, G_{LO}.
\end{equation}
We can therefore deduce the following equation for $G_{LO}$ alone:
\begin{equation}
\label{cheia}
G_{LO} = 1 + g \, {G_{LO}}^2 \left( {G_{LO}}^2 + \mu \right) \,.
\end{equation}

\begin{figure}[ht]
  \centering
 	\includegraphics[scale=0.6]{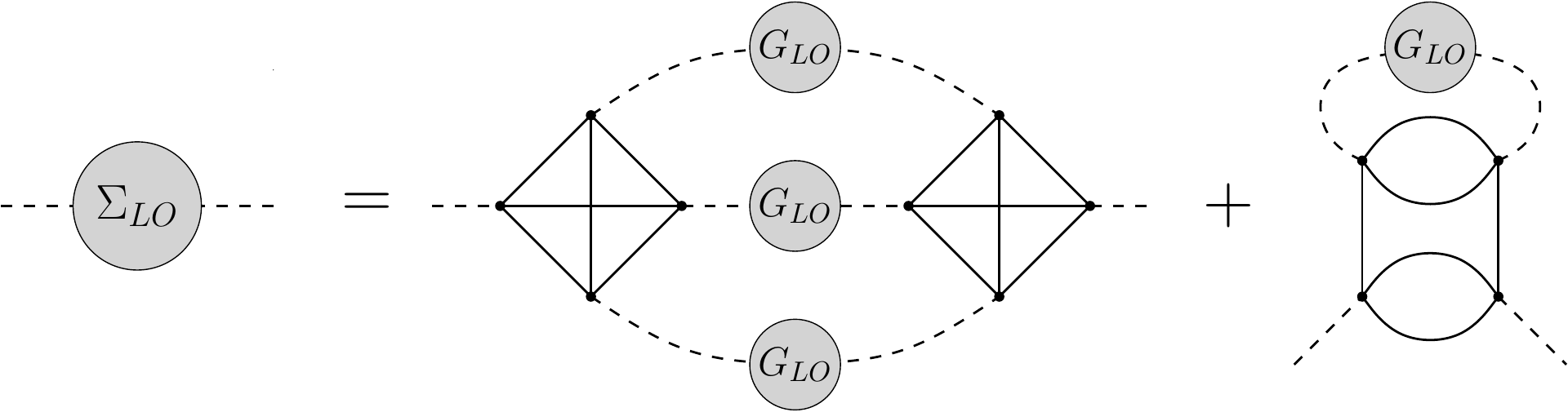}     
  \caption{Diagrammatic equation for the $2$-point function at leading order.}\label{schwinger_dyson_lo}
\end{figure}

\


Let us now investigate the NLO behaviour. The graphs contributing to the connected NLO two-point function can be obtained from the NLO vacuum graph by cutting one of the graph's internal lines. The result of this cutting process thus follows, from a combinatorial point of view, from the characterization of the NLO vacuum graphs from the previous section. One has three distinct types of graphs, depending on the type of edge that is cut in a given graph $\cG$:
\begin{enumerate}
\item one can cut an edge of a melon of type I;
\item one can cut an edge of a melon of type II;
\item one can finally cut an edge associated to a tadpole line in the core graph of $\cG$.
\end{enumerate}
\begin{figure}[ht]
  \centering
 	\includegraphics[scale=0.6]{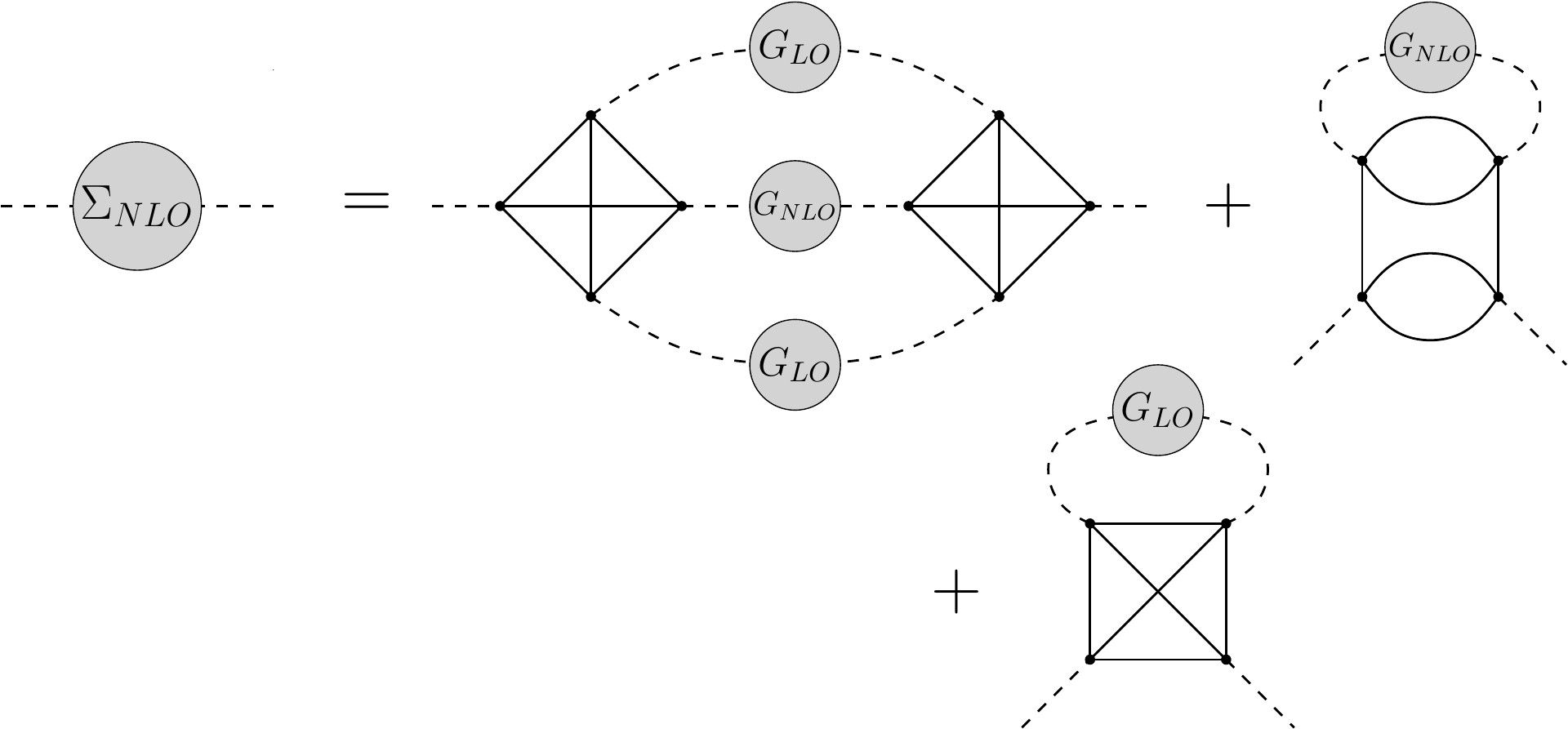}     
  \caption{Diagrammatic equation for the $2$-point function at next-to-leading order.}\label{schwinger_dyson_nlo}
\end{figure}
The diagrammatic equation obtained in this way is depicted in Figure \ref{schwinger_dyson_nlo}, and the analytic equation associated writes:
\begin{align}\label{nlo}
\Sigma_{NLO} &= 3 {\lambda_1}^2 \, {G_{LO}}^2 \, G_{NLO} - \lambda_2 \, G_{NLO} - \lambda_1 \, G_{LO} \\ 
&= 3 g \, {G_{LO}}^2 \, G_{NLO} + g \, \mu \, G_{NLO} - \sqrt{g} \, G_{LO}.
\end{align}
The combinatorial factor three in the first term above comes from the fact that one has three distinct choices for the NLO $2$-point insertion, choices corresponding to the three distinct colors of the colored edges. 

\subsection{Singularity analysis}

Let us assume for the moment that $\mu \geq 0$ and $g > 0$, so that all the series we consider have positive coefficients. By Pringsheim's theorem\footnote{Pringsheim's theorem stipulates that a power series $\underset{n\in \mathbb{N}}{\sum} a_n z^n$ with positive coefficients $a_n$ and radius of convergence $R > 0$ has a singularity at $z = R$. See e.g. \cite{flajolet}.}, we know that $G_{LO}$ is singular on the boundary of its domain of convergence.
Using standard methods of analytic combinatorics, we are furthermore able to determine how $G_{LO}$ behaves close to the critical curve\footnote{$G_{LO}$ is a generating function of tree-like objects, which lead to square-root singularities on very general grounds \cite{flajolet}.}.

\begin{proposition}\label{propo_critical} For any $\mu \geq 0$, define the critical value
\beq\label{def_gc}
g_c (\mu) = \frac{G_c (\mu) - 1}{{G_c (\mu)}^2 \left( {G_c (\mu)}^2 + \mu \right)}\,,
\eeq
where $G_c (\mu)$ is the unique real solution of the polynomial equation
\beq
- 3 \, x^3 + 4 \, x^2 - \mu \, x + 2 \mu = 0\,.
\eeq 
The (adherence of the) domain of convergence of the series defining $G_{LO}$ is $\{ (g , \mu) \in \mathbb{R}_{+} \times \mathbb{R} \,, \,   \vert g \vert \leq g_c ( \vert \mu \vert ) \}$. 
Moreover, for any $\mu \geq 0$, there exists a constant $K ( \mu ) > 0$ such that:
\begin{equation}\label{critical}
G_{LO} (g, \mu) \underset{g \to {g_c (\mu)}^{-}}{=} G_c (\mu ) - K (\mu) \sqrt{ 1 - \frac{g}{g_c (\mu)}} \left( 1 + \cO ( 1 - \frac{g}{g_c (\mu)} ) \right) \,.
\end{equation}
\end{proposition}
\begin{proof}
$G_{LO}( g, \mu )$ is a power series in $g$ with positive coefficients, therefore by Pringsheim's theorem it has a singularity at some $g_c (\mu) > 0$.
One moreover has the equation:
\begin{equation}
g = \frac{G_{LO} (g, \mu) - 1 }{ F ( G_{LO} (g, \mu) - 1 )} \equiv \Psi ( G_{LO} (g, \mu) - 1 )\,,
\end{equation}
where $F (u) := ( 1 + u)^4 + \mu (1 + u)^2$ and $\Psi$ are both analytic around $G_{LO} ( 0 , \mu ) - 1 = 0$. Since the function $g \mapsto G_{LO} (g, \mu)$ is not analytic at $g_c (\mu)$, at $\tau = G_{LO} ( g_c (\mu) , \mu ) - 1 > 0$ one must have $\Psi'(\tau) = 0$. Otherwise we could locally invert the previous equation to obtain an analytic dependence of $G_{LO}$ in a neighbourhood of $g_c (\mu)$. This leads to the equation:
\begin{equation}
F( \tau ) - \tau F' (\tau) = 0 
\,,  
\end{equation} 
known to combinatorists as the characteristic equation of the generating function $G_{LO}$ \cite{flajolet}. One can check that its unique real positive solution is $\tau = G_c (\mu ) - 1$, where $G_c (\mu)$ is inferred from the value of $\tau$ given above. In particular, $g_c(\mu)$ is indeed defined by equation (\ref{def_gc}). Moreover, the second derivative does not cancel\footnote{We use the characteristic equation to obtain this formula.}:
\beq
\Psi'' ( \tau ) = - \tau \frac{F'' ( \tau ) }{ [ F(\tau) ]^2} = - \tau \frac{ 12 (1+\tau)^2 + 2 \mu }{ [ F(\tau) ]^2} 
< 0\,.
\eeq
We therefore obtain by Taylor expansion
\begin{equation}
g_c (\mu) - g \approx - \frac{\Psi'' (\tau)}{2} \left( G_{LO} (g, \mu) - G_c (\mu ) \right)^2\,, 
\end{equation}
which can be locally inverted (by use of the singular inversion theorem) to give formula (\ref{critical}) with
\begin{equation}
K (\mu) = \sqrt{- \frac{2 g_c}{\Psi'' (\tau)}} = \sqrt{\frac{2 F(\tau)}{F''(\tau)}} = \sqrt{\frac{{G_c (\mu)}^2 \left( {G_c (\mu)}^2 + \mu \right) }{6 {G_c (\mu)}^2 + \mu}} \,.
\end{equation}

\end{proof}

\

Interestingly, the singular behaviour of $G_{LO}$ can be used to retrieve information about the asymptotics of the coefficients $C_{p,q}$ without explicitly enumerating the trees they count (as we have done in section \ref{sec:counting}). This method is central in analytic combinatorics \cite{flajolet}, and here is an example of what one can infer.

\begin{corollary} For any $\mu \geq 0$, and with the same notations as in Proposition \ref{propo_critical}, the coefficients $\alpha_n (\mu)$ of the power series $G_{LO} ( \cdot , \mu)$ behaves asymptotically as:
\begin{equation}
\alpha_n (\mu) \underset{n \to + \infty}{\sim} \frac{K( \mu ) \, g_c (\mu)^{- n}}{2 \sqrt{\pi} \,n^{3/2}}\,.
\end{equation}
\end{corollary}
\begin{proof} The analytic function $F(z) = ( 1 + z)^4 + \mu (1 + z)^2$ is aperiodic\footnote{Let us define the support $\mathrm{Supp}(F) = \{n \in \mathbb{N} \vert F_n \neq 0 \} = \{0, 1 , 2 , 3 , 4 \}$, where $F(z)= \sum_n F_n \, z^n$. $F$ being aperiodic means that there exists no $r \in \mathbb{N}$ and no integer $d \geq 2$ such that $\mathrm{Supp}(F) \subset r + d \mathbb{N}$, which is clearly the case.}. Hence one can directly apply Theorem VI.6 of \cite{flajolet} (page 405). Let us nonetheless sketch the idea of the proof. The aperiodicity of $F$ implies that of $G_{LO} ( \cdot , \mu)$, and by Daffodil's lemma (see again \cite{flajolet}, page 266), one can deduce that $G_{LO}( \cdot , \mu)$ has no other singularity than $g_c (\mu)$ on the circle $\vert g \vert = g_c ( \mu )$. The application of Cauchy's formula:
\begin{equation}
\alpha_{n} (\mu ) = \frac{1}{2 i \pi} \int_\gamma \frac{G_{LO} (z,\mu)}{z^{n + 1}} \extd z
\end{equation}
to a suitable contour $\gamma$ around $g_c (\mu)$ (known as an Hankel contour) then shows that the asymptotics of the coefficients $\alpha_n (\mu)$ is dictated by the critical behaviour at $g_c (\mu)$. Therefore the known asymptotic expansion of $\sqrt{1 - z}$ at $z=1$ directly yields the asymptotic estimate of $\alpha_n (\mu)$.
\end{proof}

In particular, taking $\mu = 1$, we obtain an estimation of the number of binary--quaternary plane trees of size $n$ (where $n$ is the number of vertices which are neither leaves nor the root). Taking $\mu = 3$ yields in turn an estimation of the number of melonic $2$-point graphs with $n$ elementary melons, which we denote by $\cM_n$. A numerical application of the previous Corollary shows that:
\begin{equation}
\cM_n \underset{n \to + \infty}{\sim} \frac{\chi \, \beta^n}{n^{3/2}}\,,
\end{equation}
with
\begin{equation}
\chi \approx 0.111 \qquad \mathrm{and} \qquad \beta \approx 14.8\,.
\end{equation}
\

Let us now briefly comment on the case $\mu <0$. Given the form of our equation for $G_{LO}$, the negative sign of $\mu$ might in principle allow for multi-critical points, that is critical points at which $\frac{\extd^2 g}{\extd {G_{LO}}^2}$ also cancels out, therefore leading to different scaling behaviours\footnote{More precisely, one would get a behaviour in $\left( 1 - \frac{g}{g_c} \right)^{1/p}$ whenever the first non-zero derivative of $g$ is of order $p \geq 2$.}. Since our definition of $G_{LO}$ is well-controlled in the region $g \leq g_c (\vert \mu \vert)$ only, we only have access to possible singularities on the boundary ($g = g_c (\vert \mu \vert)$). Exploring the phase space further would necessitate a careful study of possible analytic continuations, which is not the purpose of the present paper. We have checked, using the same method as in Proposition \ref{propo_critical}, that no new singularity is present. For instance, let us specialize to the case $\mu = - 1$, which is easier to analyse. A factor $\left( G_{LO} - 1 \right)$ can then be factored out from the relation between $g$ and $G_{LO}$, which simplifies to:
\begin{equation}
g = \frac{1}{G_{LO}^2 \, ( G_{LO} + 1 )}\,.
\end{equation}
At a critical point $(g = g_c , G_{LO} = G_c)$, $\frac{\extd g}{\extd G_{LO}}$ has to vanish, which leads to the characteristic equation:
\begin{equation}
3 {G_c}^2 + 2 G_c = 0\,.
\end{equation}
$G_c = 0$ being excluded, we conclude that $G_c = - \frac{2}{3}$. But then $g_c = \frac{27}{4} > g_c (1)$, which would bring us outside the domain of convergence of $G_{LO}$. Hence there is no singular point when $\mu = -1$, and this conclusion actually holds for arbitrary $\mu <0$.

\

We have summarized our findings on the phase space representation of Figure \ref{critical_fig}.

\begin{figure}[ht]
  \centering
 	\includegraphics[scale=0.4]{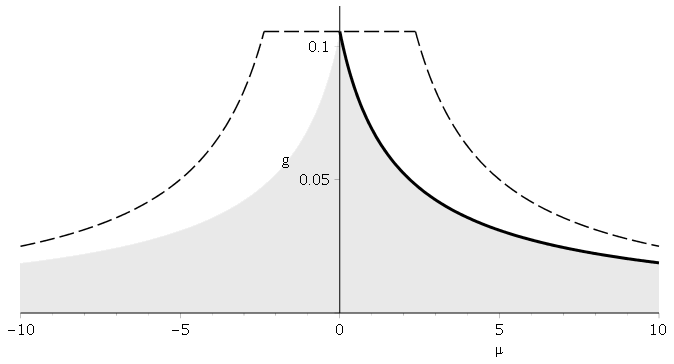}     
  \caption{Dashed line: crude bound on the domain of convergence of $G_{LO}( g , \mu)$; grey region: actual domain of convergence of $G_{LO}(g, \mu)$; black line: critical points.}\label{critical_fig}
\end{figure}

\subsection{Critical exponents}

We now use the critical behaviour of $G_{LO}$ to infer that of the free energy and deduce the value of the susceptibility critical exponent.

\subsubsection{Leading order}

We parametrize the full connected $2$-point function as:
\beq
\frac{1}{Z_N}  \int [\extd T]\, T_{i_1 i_2 i_3} \, T_{j_1 j_2 j_3} \, e^{-S_N [T]} = \frac{C_N}{N^{3/2}} \prod_{k = 1}^{3} \delta_{i_k j_k}\,,
\eeq
in such a way that
\begin{equation}
C_N =  G_{LO} + N^{-1/2} \, G_{NLO} + \ldots 
\end{equation}
The free energy itself is defined as:
\begin{equation}
F_N := \frac{1}{N^3} \ln Z_N = F_{LO} + N^{-1/2} \, F_{NLO} + \ldots
\end{equation}

The relation between $C_N$ and the leading-order free energy $F_N$ is contained in the Dyson-Schwinger equation: 
\begin{align}
	0 &= \frac{1}{Z_N} \sum_{i_1 , i_2 , i_3} \int [\extd T]\,\frac{\delta}{\delta T_{i_1 i_2 i_3}} \left( T_{i_1 i_2 i_3} \, e^{-S_N [T]} \right) \\
	&= N^3 - N^3 \, C_N + \lambda_1 \partial_{\lambda_1} \ln Z_N + 
 \lambda_2 \partial_{\lambda_2} \ln Z_N 
\end{align}
which immediately yields:
\begin{equation}\label{free_energy}
C_N (\lambda_1 , \lambda_2 ) = 1 + \left( \lambda_1 \partial_{\lambda_1} + 
 \lambda_2 \partial_{\lambda_2} \right) F_N (\lambda_1 , \lambda_2 ) \,.
\end{equation}
Extracting the leading-order contributions, and resorting to the variables $g$ and $\mu$ provides the looked for relation between $G_{LO}$ and $F_{LO}$:
\begin{equation}\label{G_F_LO}
G_{LO} (g , \mu ) = 1 + \left( 2 g \, \partial_g - \mu \, \partial_\mu \right) F_{LO} (g , \mu )\,.
\end{equation}

Close to the critical point $g_c (\mu)$, one can parametrize the most singular part of $F_{LO}$ as
\begin{equation}
K_1(\mu) \left( 1 - \frac{g}{g_c (\mu)} \right)^{2 - \gamma_{LO}}
\end{equation}
for some $K_1(\mu)$ independent of $g$. The critical exponent $\gamma_{LO}$ is the leading order susceptibility exponent and is, by equation (\ref{G_F_LO}), equal to:
\begin{equation}
\gamma_{LO}=\frac{1}{2}.
\end{equation}
This is the same critical exponent as for the $U(N)^{\otimes 3}$ invariant and MO models. This indicates that all these models have the same universal properties in the critical regime and at leading order.

\subsubsection{Next-to-leading order}

In order to compute the susceptibility exponent $\gamma_{NLO}$, one may try to directly infer the critical behaviour at next-to-leading order from the leading order one. This can be achieved by means of equation \eqref{nlo}, together with the standard QFT identity relating the connected two-point function $G_{NLO}$ to the connected two-point function $G_{LO}$ and to the $1$PI NLO two-point function $\Sigma_{NLO}$:
\begin{equation}
G_{NLO}=G_{LO}^2 \, \Sigma_{NLO}\,.
\end{equation}
Using these two equations, one gets:
\begin{equation}\label{nlo2}
G_{NLO}=\frac{ - \lambda_1 G_{LO}^3}{1+\lambda_2 \, G_{LO}^2 - 3 {\lambda_1}^2 \, G_{LO}^4} = \frac{ - \sqrt{g} \, G_{LO}^3}{1 - g \mu  G_{LO}^2 - 3 g G_{LO}^4}.
\end{equation}
Using the leading order two-point function identity \eqref{cheia}, one gets:
\begin{equation}
\frac{\partial G_{LO}}{\partial g}= \frac{G_{LO}^3 (G_{LO}^2+\mu)}{G_{LO} - 2g G_{LO}^2 ( 2 G_{LO}^2 + \mu)}.
\end{equation}
Using again identity \eqref{cheia} to express the first term of the denominator, one has:
\begin{equation}
\frac{\partial G_{LO}}{\partial g}= \frac{G_{LO}^3 (G_{LO}^2+\mu)}{1 - g \mu  G_{LO}^2 - 3 g G_{LO}^4}.
\end{equation}
One can then re-express the NLO two-point function \eqref{nlo2} as
\begin{equation}
G_{NLO}=\frac{- \sqrt{g}}{G_{LO}^2 +\mu}\frac{\partial G_{LO}}{\partial g}.
\end{equation}
We can therefore use the same argument as at leading order. First, the critical behaviour of $G_{LO}$ implies that the most singular contribution of $G_{NLO}$ is in
\begin{equation}
	\left(1 - \frac{g}{g_c (\mu)}\right)^{-1/2} \,.
\end{equation}
Second, as a consequence of relation (\ref{free_energy}), the most singular part of $F_{NLO}$ behaves as
\begin{equation}
K_2(\mu) \left(1 - \frac{g}{g_c (\mu)}\right)^{1/2} 
\end{equation} 
for some function $K_2 (\mu)$ independent of $g$.

\

We thus find the same critical value of the coupling constant (i.e. the radius of convergence) for the NLO series (as series in the coupling constant $g$) as for the leading order series. Nevertheless, one has a distinct value for the NLO susceptibility exponent:
\begin{equation}
\label{main}
\gamma_{NLO}=\frac{3}{2}\,.
\end{equation}
This again coincides with the critical exponents found in the complex and MO models. Hence we expect these three types of theory to remain in the same universality class also at next-to-leading order.




\section{Conclusion}

In this paper we have initiated the study of random tensor models with $\mO(N)^{\otimes 3}$ symmetry, which generalize both $\U(N)$ invariant and multi-orientable tensor models. Interactions are labeled by $3$-colored graphs -- not necessarily bipartite -- which represent triangulated surfaces, including non-orientable ones. We have first shown that, like $\U(N)$ invariant models, they admit a $1/N$ expansion for any finite number of non-zero coupling constants. As expected, the melonic graphs of MO and $\U(N)$ invariant models are all generated at leading order. We then focused on the quartic theory, with two types of interactions: the so-called pillow terms which represent triangulated spheres; and an interaction with the combinatorial structure of the tetrahedron, which as a colored graph however represents the projective plane (and is therefore non-orientable). We have fully characterized the leading order sector of this model, showing that it contains no more than melonic graphs. We then showed that the next-to-leading order graphs are generated by melonic $2$-point function insertions in three possible core graphs, which are nothing but colored versions of the unique core graph of the MO model. Finally, using standard techniques from analytical combinatorics, we determined the critical behaviour of the leading and next-to-leading $2$-point functions. This allowed us to reproduce the critical exponents of $\U(N)$ invariant and MO random tensor models. As a result, we may conjecture that the real model also lies in the universality class of branched polymers \cite{jr_branched}.

\medskip 

A first natural follow-up of this work is the definition of a double scaling limit for the quartic model, which amounts to simultaneously take the large $N$ limit and send the parameter $g$ to its critical value. This requires the knowledge of the critical exponents which we have computed. One might also wish to investigate further the properties of models with more interactions, which will in particular exhibit multi-critical points for specific choices of the coupling constants (see \cite{uncolored}). Finally, the colored bubbles of $\mO (N)$ invariant tensor models could be used to enlarge the theory space of tensorial group field theories, which has so far been based on $\U (N)$ invariants. This provides in particular a natural arena in which to investigate the renormalizability of MO group field theories of the type defined in \cite{MO-original}.

\bigskip

{\bf Acknowledgements:} S. C. is supported by the ANR JCJC CombPhysMat2Tens grant. A. T. is partially supported by the ANR JCJC CombPhysMat2Tens and PN 09 37 01
02 grants.

\bibliographystyle{hunsrt}
\bibliography{biblio}

\begin{thebibliography}{10}

\bibitem{razvan_colors}
Razvan Gurau.
\newblock {Colored Group Field Theory}.
\newblock {\em Commun.Math.Phys.}, 304:69--93, 2011, 0907.2582.

\bibitem{razvan_jimmy_rev}
Razvan Gurau and James~P. Ryan.
\newblock {Colored Tensor Models - a review}.
\newblock {\em SIGMA}, 8:020, 2012, 1109.4812.

\bibitem{RazvanN}
Razvan Gurau.
\newblock {The 1/N expansion of colored tensor models}.
\newblock {\em Annales Henri Poincare}, 12:829--847, 2011, 1011.2726.

\bibitem{RazvanVincentN}
Razvan Gurau and Vincent Rivasseau.
\newblock {The 1/N expansion of colored tensor models in arbitrary dimension}.
\newblock {\em Europhys.Lett.}, 95:50004, 2011, 1101.4182.

\bibitem{razvan_complete}
Razvan Gurau.
\newblock {The complete 1/N expansion of colored tensor models in arbitrary
  dimension}.
\newblock {\em Annales Henri Poincare}, 13:399--423, 2012, 1102.5759.

\bibitem{continuum_TM}
Valentin Bonzom, Razvan Gurau, Aldo Riello, and Vincent Rivasseau.
\newblock {Critical behavior of colored tensor models in the large N limit}.
\newblock {\em Nucl.Phys.}, B853:174--195, 2011, 1105.3122.

\bibitem{jr_branched}
Razvan Gurau and James~P. Ryan.
\newblock {Melons are branched polymers}.
\newblock {\em Annales Henri Poincare}, 15(11):2085--2131, 2014, 1302.4386.

\bibitem{book_ambjorn}
Jan Ambj{\o}rn, Bergfinnur Durhuus, and Thordur Jonsson.
\newblock {\em Quantum geometry: a statistical field theory approach}.
\newblock Cambridge University Press, 1997.

\bibitem{uncolored}
Valentin Bonzom, Razvan Gurau, and Vincent Rivasseau.
\newblock {Random tensor models in the large N limit: Uncoloring the colored
  tensor models}.
\newblock {\em Phys.Rev.}, D85:084037, 2012, 1202.3637.

\bibitem{ds_WDJ}
Wojciech Kamiński, Daniele Oriti, and James~P. Ryan.
\newblock {Towards a double-scaling limit for tensor models: probing
  sub-dominant orders}.
\newblock {\em New J. Phys.}, 16:063048, 2014, 1304.6934.

\bibitem{ds_SRV}
St\'{e}phane Dartois, Razvan Gurau, and Vincent Rivasseau.
\newblock {Double Scaling in Tensor Models with a Quartic Interaction}.
\newblock {\em JHEP}, 09:088, 2013, 1307.5281.

\bibitem{ds_VRJA}
Valentin Bonzom, Razvan Gurau, James~P. Ryan, and Adrian Tanasa.
\newblock {The double scaling limit of random tensor models}.
\newblock {\em JHEP}, 09:051, 2014, 1404.7517.

\bibitem{universality}
Razvan Gurau.
\newblock {Universality for Random Tensors}.
\newblock {\em Ann. Inst. H. Poincare Probab. Statist.}, 50(4):1474--1525,
  2014, 1111.0519.

\bibitem{razvan_beyond}
Razvan Gurau.
\newblock {The 1/N Expansion of Tensor Models Beyond Perturbation Theory}.
\newblock {\em Commun.Math.Phys.}, 330:973--1019, 2014, 1304.2666.

\bibitem{borel_TR}
Thibault Delepouve, Razvan Gurau, and Vincent Rivasseau.
\newblock {Universality and Borel Summability of Arbitrary Quartic Tensor
  Models}.
\newblock 2014, 1403.0170.

\bibitem{MO-original}
Adrian Tanasa.
\newblock {Multi-orientable Group Field Theory}.
\newblock {\em J. Phys.}, A45:165401, 2012, 1109.0694.

\bibitem{Tanasa-praa}
Adrian Tanasa.
\newblock {Tensor models, a quantum field theoretical particularization}.
\newblock 2012, 1211.4444.
\newblock [Proc. Rom. Acad.A13,no.3,225(2012)].

\bibitem{DRT}
Stephane Dartois, Vincent Rivasseau, and Adrian Tanasa.
\newblock {The $1/N$ expansion of multi-orientable random tensor models}.
\newblock {\em Annales Henri Poincare}, 15:965--984, 2014, 1301.1535.

\bibitem{RT}
Matti Raasakka and Adrian Tanasa.
\newblock {Next-to-leading order in the large $N$ expansion of the
  multi-orientable random tensor model}.
\newblock {\em Annales Henri Poincare}, 16(5):1267--1281, 2015, 1310.3132.

\bibitem{FT}
Eric Fusy and Adrian Tanasa.
\newblock {Asymptotic expansion of the multi-orientable random tensor model}.
\newblock 2014, 1408.5725.

\bibitem{GTY}
Razvan Gurau, Adrian Tanasa, and Donald~R. Youmans.
\newblock {The double scaling limit of the multi-orientable tensor model}.
\newblock {\em Europhys. Lett.}, 111(2):21002, 2015, 1505.00586.

\bibitem{tanasa-last}
Adrian Tanasa.
\newblock {The multi-orientable random tensor model, a review}.
\newblock 2015, 1512.02087.

\bibitem{tensor_4d}
Joseph Ben~Geloun and Vincent Rivasseau.
\newblock {A Renormalizable 4-Dimensional Tensor Field Theory}.
\newblock {\em Commun.Math.Phys.}, 318:69--109, 2013, 1111.4997.

\bibitem{tensor_3d}
Joseph Ben~Geloun and Dine~Ousmane Samary.
\newblock {3D Tensor Field Theory: Renormalization and One-loop
  $\beta$-functions}.
\newblock {\em Annales Henri Poincare}, 14:1599--1642, 2013, 1201.0176.

\bibitem{cor_u1}
Sylvain Carrozza, Daniele Oriti, and Vincent Rivasseau.
\newblock {Renormalization of Tensorial Group Field Theories: Abelian U(1)
  Models in Four Dimensions}.
\newblock {\em Commun. Math. Phys.}, 327:603--641, 2014, 1207.6734.

\bibitem{Dine_Fabien}
Dine~Ousmane Samary and Fabien Vignes-Tourneret.
\newblock {Just Renormalizable TGFT's on $U(1)^{d}$ with Gauge Invariance}.
\newblock {\em Commun. Math. Phys.}, 329:545--578, 2014, 1211.2618.

\bibitem{cor_su2}
Sylvain Carrozza, Daniele Oriti, and Vincent Rivasseau.
\newblock {Renormalization of a SU(2) Tensorial Group Field Theory in Three
  Dimensions}.
\newblock {\em Commun.Math.Phys.}, 330:581--637, 2014, 1303.6772.

\bibitem{VD}
Vincent Lahoche and Daniele Oriti.
\newblock {Renormalization of a tensorial field theory on the homogeneous space
  SU(2)/U(1)}.
\newblock 2015, 1506.08393.

\bibitem{discrete_rg}
Sylvain Carrozza.
\newblock {Discrete Renormalization Group for SU(2) Tensorial Group Field
  Theory}.
\newblock {\em Ann. Inst. Henri Poincar\'e Comb. Phys. Interact.}, 03:49--112,
  2015, 1407.4615.

\bibitem{4-eps}
Sylvain Carrozza.
\newblock {Group field theory in dimension $4-\epsilon$}.
\newblock {\em Phys. Rev.}, D91(6):065023, 2015, 1411.5385.

\bibitem{thesis}
Sylvain Carrozza.
\newblock {\em Tensorial Methods and Renormalization in Group Field Theories}.
\newblock Springer Theses, 2014, 1310.3736.

\bibitem{zinn}
Paul Zinn-Justin.
\newblock {The General O(n) quartic matrix model and its application to
  counting tangles and links}.
\newblock {\em Commun. Math. Phys.}, 238:287--304, 2003, math-ph/0106005.

\bibitem{review_marc}
Sergei Alexandrov, Marc Geiller, and Karim Noui.
\newblock {Spin Foams and Canonical Quantization}.
\newblock {\em SIGMA}, 8:055, 2012, 1112.1961.

\bibitem{review_alejandro}
Alejandro Perez.
\newblock {The Spin Foam Approach to Quantum Gravity}.
\newblock {\em Living Rev. Rel.}, 16:3, 2013, 1205.2019.

\bibitem{freidel_gft}
Laurent Freidel.
\newblock {Group field theory: An Overview}.
\newblock {\em Int.J.Theor.Phys.}, 44:1769--1783, 2005, hep-th/0505016.

\bibitem{daniele_rev2011}
Daniele Oriti.
\newblock {The microscopic dynamics of quantum space as a group field theory}.
\newblock pages 257--320, 2011, 1110.5606.

\bibitem{Vince_2d}
Andrew Vince.
\newblock The classification of closed surfaces using colored graphs.
\newblock {\em Graphs and Combinatorics}, 9(1):75--84, 1993.

\bibitem{FL}
Laurent Freidel and David Louapre.
\newblock {Nonperturbative summation over 3-D discrete topologies}.
\newblock {\em Phys. Rev.}, D68:104004, 2003, hep-th/0211026.

\bibitem{karim}
Jacques Magnen, Karim Noui, Vincent Rivasseau, and Matteo Smerlak.
\newblock {Scaling behaviour of three-dimensional group field theory}.
\newblock {\em Class. Quant. Grav.}, 26:185012, 2009, 0906.5477.

\bibitem{Ryan_jacket}
James~P. Ryan.
\newblock {Tensor models and embedded Riemann surfaces}.
\newblock {\em Phys. Rev.}, D85:024010, 2012, 1104.5471.

\bibitem{freidel_louapre}
Laurent Freidel and David Louapre.
\newblock {Nonperturbative summation over 3-D discrete topologies}.
\newblock {\em Phys. Rev.}, D68:104004, 2003, hep-th/0211026.

\bibitem{malek}
Abdelmalek Abdesselam.
\newblock {The Jacobian conjecture as a problem of perturbative quantum field
  theory}.
\newblock {\em Annales Henri Poincare}, 4:199--215, 2003, math/0208173.

\bibitem{andrea}
Axel de~Goursac, Andrea Sportiello, and Adrian Tanasa.
\newblock {The Jacobian Conjecture, a Reduction of the Degree to the Quadratic
  Case}.
\newblock 2014, 1411.6558.

\bibitem{flajolet}
Philippe Flajolet and Robert Sedgewick.
\newblock {\em Analytic Combinatorics}.
\newblock Cambridge University Press, 2009.

\end{thebibliography}
%
%
%

\end{document}